%% file: main.tex
\def\isarxivversion{1} %%% for icml submission version, we comment this line

\ifdefined\isarxivversion
\documentclass[11pt]{article}
\else
\documentclass{article}
\usepackage[accepted]{icml2020} %%% Zhao: this is for accepted version
\fi

\usepackage{graphicx}
\usepackage{comment}
\usepackage{amsmath}
\usepackage{amsthm}
\usepackage{amssymb}
\usepackage{float}
\usepackage{subfig}
\usepackage{algorithm}
\usepackage{algpseudocode}
\usepackage{tikz}
\usepackage{hyperref}  
\usepackage{comment}
\usepackage{multirow}
\usepackage{booktabs}

\hypersetup{colorlinks=true,citecolor=red,linkcolor=blue} 
\usetikzlibrary{arrows}

\graphicspath{{./figs/}}

\ifdefined\isarxivversion
\usepackage[margin=1in]{geometry}
\else 

\fi

\ifdefined\isarxivversion
\else

\fi

\newtheorem{theorem}{Theorem}[section]
\newtheorem{lemma}[theorem]{Lemma}
\newtheorem{definition}[theorem]{Definition}

\newtheorem{corollary}[theorem]{Corollary}
\newtheorem{conjecture}[theorem]{Conjecture}

\newtheorem{hypothesis}[theorem]{Hypothesis}

\renewcommand{\tilde}{\widetilde}
\renewcommand{\hat}{\widehat}

\DeclareMathOperator{\poly}{poly}

\DeclareMathOperator{\R}{{\mathbb R}}

\DeclareMathOperator*{\E}{{\mathbb{E}}}

\definecolor{b2}{RGB}{51,153,255}
\definecolor{mygreen}{RGB}{80,180,0}
\definecolor{yl}{RGB}{255,80,0}

\newcommand{\mixup}{{\em Mixup}}
\newcommand{\instahide}{{\em InstaHide}}

\newcommand{\tabincell}[2]{\begin{tabular}{@{}#1@{}}#2\end{tabular}}

\begin{document}

\ifdefined\isarxivversion

\title{{\instahide}: Instance-hiding Schemes for Private Distributed Learning\thanks{A preliminary version of this paper appeared in the Proceedings of the 37th International Conference on Machine Learning (ICML 2020).}}

\date{}

\author{
Yangsibo Huang\thanks{\texttt{yangsibo@princeton.edu}. Princeton University.}
\and
Zhao Song\thanks{\texttt{zhaos@princeton.edu}. Princeton University and Institute for Advanced Study.}
\and
Kai Li\thanks{\texttt{li@cs.princeton.edu}. Princeton University.}
\and
Sanjeev Arora\thanks{\texttt{arora@cs.princeton.edu}. Princeton University and Institute for Advanced Study.}
}

\else

\icmltitlerunning{{\it InstaHide}: Instance-hiding Schemes for Private Distributed Learning}

\twocolumn[
\icmltitle{{\instahide}: Instance-hiding Schemes for Private Distributed Learning}

\icmlkeywords{Machine Learning, ICML}

\vskip 0.3in
]

\fi

\ifdefined\isarxivversion
\begin{titlepage}
\maketitle
\begin{abstract}
\input{abstract}
\end{abstract}
\thispagestyle{empty}
\end{titlepage}

\else

\begin{abstract}
\input{abstract}
\end{abstract}

%%%% Zhao: below six lines are making sure that we can have page numbers.
\iffalse
\makeatletter
\let\ps@oldempty\ps@empty % save default definition of \ps@empty
\renewcommand\ps@empty\ps@plain
\makeatother 
\pagestyle{plain}
\setcounter{page}{1}
\fi
%%%% Zhao : Done :)

\fi

\input{intro}   %%%% Section 1. Introduction
\input{mixup}   %%%% Section 2. Mixup Method and Some Attacks
\input{alg}     %%%% Section 3. InstaHide
\input{hard}    %%%% Section 4. Security Analaysis
\input{exp}     %%%% Section 5. Experiments
                %%%% Section 6. InstaHide Deployment: Best practice
                %%%% Section 7. A Challenge Dataset
\input{relat}   %%%% Section 8. Related Work
\input{concl}   %%%% Section 9. Discussion of Potential Attacks
                %%%% Section 10. Conclusion

%\newpage
\section*{Acknowledgments}
This project is supported in part by Princeton University fellowship, Ma Huateng Foundation, Schmidt Foundation, Simons Foundation, NSF, DARPA/SRC, Google and Amazon AWS. Arora and Song were at the Institute for Advanced Study during this research. 

We would like to thank Amir Abboud, Josh Alman, Boaz Barak, and Hongyi Zhang for helpful discussions, and
Mark Braverman, Matthew Jagielski, Florian Tramèr, Nicholas Carlini and his team for suggesting  attacks.

\newpage
{%\small
\bibliographystyle{alpha}
\bibliography{ref}
}

\onecolumn
\appendix

\section*{Appendix}
\input{hide_app}
%%%% Section A. discuss about InstaHide

\input{attack_app}
%%%% Section B. discuss about possible attack on mixup

\input{hard_app}

%%%% Section C. discuss more about hardness result, e.g., 3-SUM, k-SUM

\input{phase_app}

%%%% Section D. discuss more about compressive sensing and phase retrieval

\input{exp_app}
%%%% Section E. discuss more about experiments

\end{document}

%% file: abstract.tex
How can multiple distributed entities collaboratively train a shared deep net on their private data while preserving  privacy? This paper introduces {\em InstaHide}, a simple encryption of training images, which can be plugged into  existing distributed deep learning pipelines.  The encryption is efficient and applying it during training has minor effect on test accuracy. 

{\em InstaHide} encrypts each training image with a ``one-time secret key'' which consists of mixing a number of randomly chosen images and applying a random pixel-wise mask.  Other contributions of this paper include:
(a) Using a large public dataset (e.g. ImageNet) for mixing during its encryption, which improves security. 
(b) Experimental results to show effectiveness in preserving privacy against known attacks with only minor effects on accuracy. 
(c) Theoretical analysis showing that successfully attacking privacy requires attackers to solve a difficult computational problem. 
(d) Demonstrating that use of the pixel-wise mask is important for security, since {\em Mixup} alone  is  shown to be insecure to some some efficient attacks.  
(e) Release of a challenge dataset\footnote{ \href{https://github.com/Hazelsuko07/InstaHide_Challenge}{https://github.com/Hazelsuko07/InstaHide\_Challenge}.} to encourage new attacks. Our code is available at \href{https://github.com/Hazelsuko07/InstaHide}{https://github.com/Hazelsuko07/InstaHide}. 

%% file: intro.tex
\section{Introduction}
\label{sec:intro}

In many applications, multiple parties or clients with sensitive data want to collaboratively train a neural network. For instance, hospitals may wish to train a model on their patient data. However, aggregating data to a central server may violate regulations such as Health Insurance Portability and Accountability Act (HIPAA)~\cite{act1996health} and General Data Protection Regulation (GDPR)~\cite{voigt2017eu}.

Federated learning~\cite{mcmahan2016communication, kmyrsb16} proposes letting participants train on their own data in a distributed fashion and share only model updates ---i.e., gradients---with the central server. The server aggregates these updates (typically by averaging) to improve a global model and then sends updates to participants.  This process runs iteratively until the global model converges. Merging information from individual data points into aggregated gradients intuitively preserves privacy to some degree. On top of that, it is possible to add noise to gradients in accordance with Differential Privacy (DP)~\cite{dkmmn06,dr14}, though  careful calculations are needed to compute the amount of noise to be added~\cite{acg+16,dpsgd19}. However, the privacy guarantee of DP only applies to the trained model (i.e., approved use of data) and does not apply to side-channel computations performed by curious/malicious parties who are privy to the communicated gradients. Recent work~\cite{zlh19} suggests that eavesdropping attackers can recover private inputs from shared model updates, even when DP was used. A more serious issue with DP is that meaningful guarantees involve adding so much noise that test accuracy reduces by over $20\%$ even on CIFAR-10~\cite{dpsgd19}.

Cryptographic methods such as {\em secure multiparty computation} of~\cite{yao82} and fully-homomorphic encryption~\cite{g09} can ensure privacy against arbitrary side-computations by adversary during training.   Unfortunately it is a challenge to use them in modern deep learning settings, owing to their high computational overheads and their needs for special setups (e.g finite field arithmetic, public-key infrastructure).

Here we introduce a new  method  {\em InstaHide},  inspired by a weaker cryptographic idea of {\em instance hiding} schemes~\cite{afk87}. 
We only apply it to image data in this paper and leave other data types (e.g., text) for future work. {\instahide} gives a way to transform input $x$ to a hidden/encrypted input $\tilde{x}$ in each epoch such that: (a) Training deep nets using the $\tilde{x}$'s instead of $x$'s gives nets almost as good in terms of final accuracy; (b) Known methods for recovering information about $x$ out of $\tilde{x}$ are computationally very expensive. In other words, $\tilde{x}$ effectively hides information contained in $x$ except for its label. 

{\instahide} encryption has two key components. The first is inspired by  {\em Mixup} data augmentation method~\cite{zcdl17}, which trains deep nets on composite images created via linear combination of pairs of images (viewed as vectors of pixel values). 
In {\instahide} the first step when encrypting image $x$ (see Figure~\ref{fig:main_fig}) is to take its linear combination with $k-1$ randomly chosen images from either the participant's private training set or from a large public dataset (e.g., ImageNet~\cite{imagenet09}).  The second step of {\instahide} involves applying a random pattern of sign flips on the pixel values of this composite image, yielding encrypted image $\tilde{x}$, which  can be used as-is in existing  deep learning frameworks. Note that the set of random images for mixing and the random sign flipped mask are used only once ---in other words, as a one-time key that is  never re-used for another encryption.

The idea of random sign flipping\footnote{\label{ft:sign_flip}Note that randomly flipping signs of coordinates in vector $x$ can be viewed alternatively as retaining only absolute value of each pixel in the mixed image. If a pixel value was $c$ in the original image, the mixup portion of our encryption amounts to scaling it by $\lambda$ and adding some value $\eta$ to it from the mixing images, and the random sign flip amounts to retaining $|\lambda c+\eta|$. Figure~\ref{fig:GAN_demask} gives an illustration.

We choose to take the viewpoint of random sign flips because this is more useful in extensions of {\instahide}, which are forthcoming.}
%---which is a one-time key, or {\em nonce}, never reused for another encryption during the protocol---
is inspired by  classic {\em Instance-Hiding} over finite field ${\sf GF}(2)$, which involves adding a random vector $r$  to an  input $x$. (See Appendix~\ref{sec:hide_app} for background.) Adding $1$ over ${\sf GF}(2)$ is analogous to a sign flip over $\mathbb{R}$. (Specifically, the groups $({\sf GF}(2), +)$ and $(\{\pm 1\}, \times)$ are isomorphic.) The use of a public dataset in {\instahide} plays a role reminiscent of {\em random oracle} in cryptographic schemes~\cite{canetti2004random} ---the larger this dataset, the better the conjectured security level (see Section~\ref{sec:privacy}). A large private dataset would suffice too for security, but then would require prior coordination/sharing among participants.
% to coordinate 

\begin{figure*}[t]
    \centering
    \includegraphics[width=0.99\textwidth]{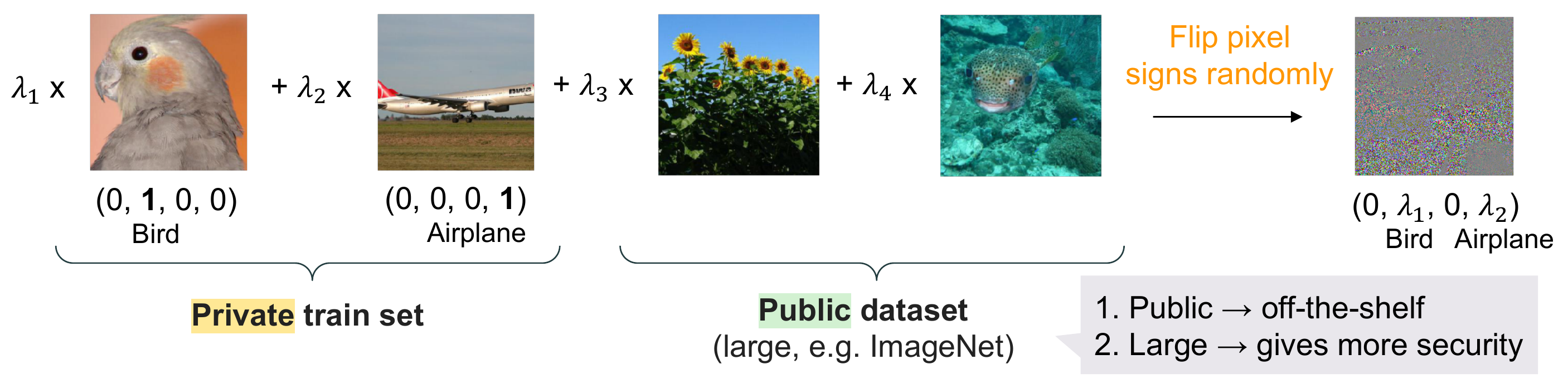}
    \caption{Applying {\instahide} ($k=4$) to the leftmost private image consists of mixing with another (private)  image randomly chosen from the training set and two images randomly chosen from a fixed large public dataset. This is  followed by a random sign-flipping mask on the composite image. To encrypt another image, different random choices will get used (``one-time key").}
    \label{fig:main_fig}
    
\end{figure*}

Experiments on MNIST, CIFAR-10, CIFAR-100 and ImageNet datasets (see Section~\ref{sec:exp}) suggest that {\instahide} is an effective approach to hide training images from attackers. It is much more effective at hiding images than {\mixup} alone and provides  better trade-off between privacy preservation and accuracy than DP. To enable further rigorous study of attacks, we release a challenge dataset of images encrypted using {\instahide}.

\paragraph{Enhanced functionality due to {\instahide}.}
As hinted above, {\instahide} plugs seamlessly into existing distributed learning frameworks such as federated learning: clients encrypt their inputs on the fly with {\instahide} and participate in training (without using DP). Depending upon the level of security needed in the application (see Section~\ref{sec:security_part1}), {\instahide} can also be used to present enhanced functionality that are unsafe in current distributed frameworks. For instance, in each epoch, computationally limited clients can encrypt each private input $x$ to $\tilde{x}$ and ship it to the central server for all subsequent computation. The server may randomize in the pooled data to create its own batches to deal with special learning situations when distributed data are not independent and identically distributed.

\paragraph{Security goal of {\instahide}.} 
%{\em InstaHide} is \textbf{not} intended to be a mission-critical encryption like RSA \cite{rivest1978method}; it is designed to give users a light-weight encryption method that allows them to use machine learning without giving eavesdroppers or servers access to their raw data. To the best of our knowledge, there is no other cost-effective alternative to {\em InstaHide} for this application.

The goal of {\em InstaHide} is to provide a light-weight encryption method to make it difficult for attackers to recover the training data in a large training dataset in a distributed learning setting, with minor reduction on data utility.  It is not 
designed to provide the level of security strength as encryption methods such as RSA \cite{rivest1978method}. It is an initial step towards exploring better privacy preservation while maintaining data utility. 

\paragraph{Rest of the Paper:} Section~\ref{sec:warmup} recaps {\mixup} and suggests it alone is not secure. Section~\ref{sec:instahide} presents two {\instahide} schemes, and Section~\ref{sec:privacy} analyzes their security\footnote{A recent attack \cite{carlini_attack} (see Section \ref{sec:potential_attacks} for details) suggests that the security analysis for {\instahide} in the current version may need some revisions. We will release an update with more extensive explanations.}. Section~\ref{sec:exp} shows experiments for {\instahide}'s efficiency, efficacy, and security. Section~\ref{sec:privacy_suggestion} provides suggestions for practical use, and Section~\ref{sec:challenge_dataset} describes the challenge dataset. We review related work in Section~\ref{sec:related}, discuss potential attacks in Section~\ref{sec:potential_attacks} and conclude in Section~\ref{sec:conclude}.

%% file: mixup.tex
\section{{\mixup} and Its Vulnerability}
\label{sec:warmup}

This section reviews {\mixup} method~\cite{zcdl17} for data augmentation in deep learning and shows ---using two plausible attacks--- that it alone does not assure privacy.  For ease of description, we consider a vision task with a private dataset $\mathcal{X} \subset \R^d$ of size $n$. Each image $x_i \in \mathcal{X}, i \in [n]$ is normalized such that $\sum_{j=1}^{d} x_{i,j} = 0$ and $\|x_i\|_2 = 1$.

See Algorithm \ref{alg:mixup}. Given an original dataset, in each epoch of the training the algorithm generates a {\mixup} dataset on the fly, by linearly combining $k$ random samples, as well as their labels (lines~\ref{line:mix_x}, \ref{line:mix_y}). {\mixup} suggests that training with mixed samples and mixed labels serves the purpose of data augmentation, and achieves better test accuracy on normal images.

We describe the algorithm as an operation on $k$ images, but  previous works mostly used $k = 2$.

\subsection{Attack on {\mixup} within a Private Dataset }
\label{sec:mixup_inside_attack}
 
We propose an attack to the {\mixup} method when a private image is mixed up more than once during training. 
In fact, in Algorithm~\ref{alg:mixup}, each image is used $kT$ times during training, where $k$ is the number of samples to mix, and $T$ is the number of training epochs.

\begin{algorithm}[t]\caption{{\mixup}~\cite{zcdl17}}
\small
    \begin{algorithmic}[1]
    \Procedure{\textsc{Mixup}}{$W,T,\mathcal{X},\mathcal{Y}$} 
        \State $W$: the weights of the deep neural network; $T$: number of epochs; ${\cal X} = \{ x_1, \cdots, x_n \}, {\cal Y} = \{ y_1, \cdots, y_n \}$: the original dataset.
        \State Initialize $W$ 
        \For{$t = 1 \to T$} 
            \State Generate $\pi_1 $ such that $\pi_1(i) = i$, $\forall i \in [n]$, and $k-1$ random permutations $\pi_2$, $\cdots$, $\pi_k : [n] \rightarrow [n]$ \Comment{$[n]$ denotes $\{1,2,\cdots,n\}$}
            \State Sample $\lambda_1, \cdots, \lambda_n \sim [0,1]^k $ uniformly at random, and for all $i\in [n]$ normalize $\lambda_i$ such that $\| \lambda_i \|_1 = 1$.
            \State $\tilde{D} \gets \emptyset$
                \For{$i = 1 \to n$} \Comment{Generate {\mixup} dataset}
                        \State $\tilde{x}_i^{\mathrm{mix}} \leftarrow \sum_{j=1}^{k} (\lambda_i)_j \cdot x_{ \pi_j( i )}$ \Comment{Mix images} \label{line:mix_x}
                        \State $\tilde{y}_i \leftarrow \sum_{j=1}^{k} (\lambda_i)_j \cdot y_{\pi_j( i )}$ \Comment{Mix labels} \label{line:mix_y}
                        \State $\tilde{D} \gets \tilde{D} \cup (\tilde{x}_i, \tilde{y}_i)$
                \EndFor
                \State Train $W$ using the {\mixup} dataset $\tilde{D}$
        \EndFor
\EndProcedure
\end{algorithmic}
\label{alg:mixup}
\end{algorithm}

Assume that pairs of images in $\mathcal{X}$ are fairly independent at the pixel level, so that the
inner product of a random pair of images (viewed as vectors)  has expectation $0$ (expectation can be nonzero but small).  We can think of each pixel is generated from some distribution with standard deviation $1/\sqrt{d}$ and describe attacks with this assumption. (Section~\ref{sec:exp} presents experiments showing that these attacks do work in practice.)

Suppose we have two {\mixup} images $\tilde{x}_1$ and $\tilde{x}_2$ which are derived from two subsets of private images, ${\mathcal{S}}_1, {\mathcal{S}}_2 \subset \mathcal{X}$ and $|{\mathcal{S}}_1| = |{\mathcal{S}}_2| = k$. If $\tilde{x}_1$ and $\tilde{x}_2$ contain different private images, namely ${\mathcal{S}}_1 \cap {\mathcal{S}}_2 = \emptyset$, then the expectation of $\tilde{x}_1 \cdot \tilde{x}_2$ is 0. However, if ${\mathcal{S}}_1 \cap {\mathcal{S}}_2 \neq \emptyset$, and $\tilde{x}_1$ and $\tilde{x}_2$ have coefficients $\lambda_1$ and $\lambda_2$ for the common image in these two sets, then the expectation of $\langle \tilde{x}_1 , \tilde{x}_2 \rangle$ is $\lambda_1 \lambda_2/k$, which means by simply checking the inner products between two $\tilde{x}$'s, the attacker can determine with high probability whether they are derived from the same image. Thus if the attacker finds multiple such pairs, they can average the $\tilde{x}$'s to start getting a  good estimate of $x$. (Note that the rest of images in the pairs are with high probability distinct and so average to  $0$.)

\subsection{Attack on {\mixup} Between a Private and a Public Dataset}
\label{sec:mixup_attack_public}

To defend against the previous attack, 
it seems that a possible method is to modify the {\mixup} method to mix a private image $x$ with $k-1$ images only once to get a single $\tilde{x}$, and use this $\tilde{x}$ as surrogate for $x$ in all epochs.
To ensure $x \in \mathcal{X}$ is used only once, it uses an additional public dataset $\mathcal{X'}$ (e.g. ImageNet). 
In other words, for every $x \in \mathcal{X}$, it produces $\tilde{x}$ by using {\mixup} between $x$ and $k-1$ random images from a large public dataset.  

This extension of the {\mixup} method seems secure at first glance, as naively one can imagine that to violate privacy, the adversary must do exhaustive search over $(k-1)$-tuples of 
public images to determine which were mixed into $\tilde{x}$, and try all possible $k$-tuples of coefficients, and then subtract the corresponding sum from $\tilde{x}$ to extract $x$. If this is true, it would suggest that extracting $x$ or any approximation to it requires ${N \choose {k-1}} \approx N^{k-1}$ work, where $N$ is the number of images in the public dataset. This work becomes infeasible even for $k=4$. 
However, we sketch an attack below that runs in $O(Nk)$ time. 

It again uses the above assumption about the pairwise independence property of a random image pair. Recall that standard deviation of pixels is $ 1 / \sqrt{d}$. Namely, to determine the images that went into the mixed sample $\tilde{x} = \lambda_1 x_1 + \sum_{i=2}^k \lambda_i x_i$, it suffices to go through each image $z$ in the dataset and examine the inner product $z\cdot \tilde{x}$. If $z$ is not one of the $x_i$'s then this inner product is 
of the order at most $\sqrt{k/d}$ (see part 1 of Theorem~\ref{thm:attack_in_X}), whereas if 
it is one of the $x_i$'s then it is of the order
at least $\frac{1}{k}(1 - \sqrt{k/d} )$ (see part 2 of Theorem~\ref{thm:attack_in_X}).
Thus if $k^3 \ll d$ (which is true if the number of pixels $d$ is a few thousand) then the inner product gives a strong signal whether $z$ is one of the $x_i$'s.
Once the correct $x_i$'s and their coefficients have been guessed, we obtain $x$ up to a linear scaling. Thus the above attack works with good probability and in time proportional to the size of the dataset. We provide results for this attack in Section~\ref{sec:exp}.

This attack of course requires that $x$ is being mixed in with images from a public dataset $\mathcal{X'}$. In the case that $\mathcal{X'}$ is private and diverse enough, it is conceivable that  {\mixup} is safe. (MNIST for example may not be diverse enough but ImageNet probably is.) We leave it as an open question. 

\subsection{Discussions}
\begin{table}[t]
\setlength{\tabcolsep}{1mm}
% \small
    \centering
    \begin{tabular}{ll}
        \toprule
         \textbf{Setting} & \textbf{Secure?}  \\ 
         \midrule
        $x$ is mixed in multiple $\tilde{x}$'s & No \\ 
        $x$ is mixed in a single $\tilde{x}$, with a public dataset & No  \\
        $x$ is mixed in a single $\tilde{x}$, with a private dataset & Maybe \\
        \bottomrule
      \end{tabular}
      \caption{Security of {\mixup} alone.}
    \label{tab:mixup_insecure}
    
\end{table}

Table~\ref{tab:mixup_insecure} summarizes the security of {\mixup} . Only  {\mixup} with single $\tilde{x}$ for each private $x$ and within a private dataset(s) has not been identified vulnerable to potential attacks. But this does not necessarily mean it is secure. In addition, the test accuracy in this case is not comparable to vanilla training; it incurs about $20\%$ accuracy loss with CIFAR-10 tasks.

\paragraph{Potentially insecure applications of {\mixup}.} Recently, a method called FaceMix ~\cite{facemix19} applies {\mixup} at the representation level (i.e. intermediate output of deep models) during inference. Given a representation function $h: \R^d \rightarrow \R^l$ and a secret $x\in \R^d$,  the paper assumed a threat model that the attacker is able to reconstruct $x$ given $h(x)$. Therefore, FaceMix proposed to protect the privacy of $x$ by generating $\tilde{h}_1(x)= \lambda_1 h(x) + \sum_{i=2}^k \lambda_i h(x_i)$, and run inference on $\tilde{h}(x)$. A similar but different idea was shown in \cite{fwxmw19}, which proposed to generate $\tilde{x} = \lambda_1 x + \sum_{i=2}^k \lambda_i x_i$ and use $\tilde{h}_2(x)= h(\tilde{x})$ for training. 

Both methods may be vulnerable to the attacks presented in this section: when $h$ is linear (one example in \cite{facemix19}), we have $\tilde{h}_1(x) = \tilde{h}_2(x)$, which means the attacker can reconstruct the {\mixup} image $\tilde{x} = \lambda_1 x + \sum_{i=1}^k x_i$ from $\tilde{h}_1(x)$ and run attacks on {\mixup}. For a nonlinear $h$, $\tilde{h}_1(x) \approx \tilde{h}_2(x)$ may also hold.

%% file: alg.tex
\section{{\instahide}}\label{sec:instahide}

This section first presents two schemes: Inside-dataset {\instahide} and Cross-dataset {\instahide}, and then describes their inference.

The Inside-dataset {\instahide} mixes each training image with random images within the same private training dataset.  The cross-dataset {\instahide}, arguably more secure (see Section~\ref{sec:privacy}), involves mixing with random images from a large public dataset like ImageNet.

\subsection{Inside-Dataset {\instahide}}\label{sec:pixel_wise}

Algorithm \ref{alg:instahide} shows Inside-dataset {\instahide}. Its encryption step includes mixing the secret image $x$ with $k-1$ other training images followed by an extra random pixel-wise sign-flipping mask on the composite image. Note that random sign flipping changes the color of a pixel. The motivation for random sign flips was described around Footnote~\ref{ft:sign_flip}.

\vspace{1mm}
\begin{definition}[random mask distribution $\Lambda_{\pm}^d$]
    \label{def:lam_pm}
    Let $\Lambda_{\pm}^d$ denote the $d$-dimensional random sign distribution such that $\forall \sigma \sim \Lambda_{\pm}^d$,  for $i \in [d]$, $\sigma_i$ is independently chosen from $\{\pm1\}$ with probability 1/2 each.
\end{definition}

\begin{algorithm}[t]\caption{Inside-dataset {\instahide}}\label{alg:instahide}
        \small
        \begin{algorithmic}[1]
        \Procedure{\textsc{InstaHide}}{$W, T, \mathcal{X},\mathcal{Y}$} \Comment{This paper} 
            \State $W$: weights of the neural network;  $T$: number of epochs;
            \newline\indent
            ${\cal X} = \{ x_1, \cdots, x_n \}$: data;
            ${\cal Y} = \{ y_1, \cdots, y_n \}$: labels.
            \State Initialize $W$ 
            \For{$t = 1 \to T$} 
                \State Generate $\pi_1$ such that $\pi_1(i)=i$, $\forall i \in [n]$, and $k-1$ random permutations $\pi_2$, $\cdots$, $\pi_k: [n] \rightarrow [n]$
                \State Sample $\lambda_1, \cdots, \lambda_n \sim [0,1]^k $ uniformly at random, and for all $i \in [n]$ normalize $\lambda_i \in \R^k$ such that $\| \lambda_i \|_1 = 1$ and $\| \lambda_i \|_\infty \leq c_1$. \Comment{$c_1 \in [0,1]$ is a constant that upper bounds a single coefficient} \label{lin:sample_lam}
                \State Sample $\sigma_1, \cdots, \sigma_n \sim \Lambda_{\pm}^d $ uniformly at random.  \Comment{Definition~\ref{def:lam_pm}}
                \State $\tilde{D} \gets \emptyset$
                \For{$i = 1 \to n$} \Comment{Generate {\instahide} dataset}
                    \State $\tilde{x}_i \leftarrow \sigma_i \circ \sum_{j=1}^{k} (\lambda_i)_j \cdot x_{ \pi_j( i )}$ \Comment{Encryption} \label{lin:compute_wt_x_i}
                    \State $\tilde{y}_i \leftarrow \sum_{j=1}^{k} (\lambda_i)_j \cdot y_{\pi_j( i )}$ \Comment{Mix labels}
                    \State $\tilde{D} \gets \tilde{D} \cup (\tilde{x}_i, \tilde{y}_i)$
                \EndFor 
                \State Train $W$ using the {\instahide} dataset $\tilde{D}$
            \EndFor
        \EndProcedure
        \end{algorithmic}
\end{algorithm}

To encrypt a private input $x_i$, we first determine the random coefficient $\lambda$'s for image-wise combination, but with the constraint that they are at most $c_1$ to avoid dominant leakage of any single image (line~\ref{lin:sample_lam} in Algorithm~\ref{alg:instahide}).  Then we sample a random mask $\sigma_i \sim \Lambda_\pm^d$ and apply $\sigma \circ x$, where $\circ$ is coordinate-wise multiplication of vectors (line~\ref{lin:compute_wt_x_i} in Algorithm~\ref{alg:instahide}). Note that the random mask $\sigma_i$ and the $k-1$ images used for mixing with $x_i$, will not be reused to encrypt other images.  They constitute a ``random one-time private key.''

\paragraph{\em A priori.} It may seem that using a different mask for each training sample would completely destroy the accuracy of the trained net, but as we will see later it has only a small effect when $k$ is small. Mathematically, this seems reminiscent of the {\em phase retrieval problem}~\cite{csv13,ln18} (see Appendix~\ref{sec:app_phase_retrieval}).

\subsection{Cross-Dataset {\instahide}}
\label{sec:cross_dataset}

Cross-dataset {\instahide} extends the encryption step of Algorithm~\ref{alg:instahide} by mixing $k$ images from the  private training dataset $D_{\text{private}}$ and a public dataset $D_{\text{public}}$, and a random mask as a random one-time secret key.  

Although the second dataset can be private, there are several motivations to 
use a public dataset: (a) Some privacy-sensitive datasets, (e.g. CT or MRI scans), feature images with certain structure patterns with uniform backgrounds. Mixing among such images as in Algorithm~\ref{alg:instahide} would not hide information  effectively. 
(b) Drawing mixing images from a larger dataset gives greater unpredictability, hence better security (see Section~\ref{sec:privacy}). 
(c) Public datasets are freely available and eliminate the need for special setups among participants in a distributed learning setting.

To mix $k$ images in the encryption step, we randomly choose $2$ images from $D_{\text{private}}$ and the other $k-2$ from $D_{\text{public}}$, and apply  {\instahide} to all these images. The only difference in the cross-dataset scheme is that, the model is trained to learn {\em only} the (mixed) label of $D_{\text{private}}$ images. We assume $D_{\text{public}}$ images are unlabelled. For better accuracy, we lower bound the sum of coefficients of two private images by a constant $c_2 \in [0, 1]$.

We advocate preprocessing a public dataset in two steps to obtain $D_{\text{public}}$ for better security.  The first is to randomly crop a number of patches from each image in the public dataset to form $D_{\text{public}}$.  This step will make $D_{\text{public}}$ much larger than the original public dataset.  The second is to filter out the ``flat'' patches.  In our implementation, we design a filter using SIFT~\cite{lowe1999object}, a feature extraction technique to retain patches with more than 40 key points.

\subsection{Inference with {\instahide}}
\label{sec:inference}

Either scheme above by default applies {\instahide} during inference, by averaging predictions of multiple encryptions (e.g. 10) of a test sample. This idea is akin to existing cryptographic frameworks for secure evaluation on a public server via homomorphic encryption (e.g.~\cite{mishra2020delphi}).
Since the encryption step of 
{\instahide} is very efficient,
the overhead of such inference is quite small.

One can also choose not to apply 
{\instahide} during inference.  We found in our experiments (Section~\ref{sec:exp}) that it works for low-resolution image datasets such as CIFAR-10 but it does not work well with a high-resolution image dataset such as ImageNet.

%% file: hard.tex
\section{Security Analysis}
\label{sec:privacy}

This section considers the security of {\instahide} in distributed learning, specifically the 
cross-dataset version.

\noindent{\bf Attack scenario:}
In each epoch, all clients replace each (image, label)-pair $(x, y)$ in the training set with some $(\tilde{x}, \tilde{y})$ using {\instahide}.
Attackers observe $h(\tilde{x}, \tilde{y})$ for  some function $h$:  in federated learning  $h$ could involve batch gradients or hidden-layer activations computed using input $\tilde{x}, \tilde{y}$ as well as other inputs. 

Argument for security consists of two halves: (1) To recover significant information about an image $x$ from communicated information, {\em computationally limited}  eavesdroppers/attackers have to break {\instahide} encryption (Section~\ref{sec:security_part1}). (2) Breaking {\instahide} is difficult (Section~\ref{sec:security_part2}). 

\begin{figure*}[t]
    \centering
    \subfloat{
    \includegraphics[width=0.29\textwidth]{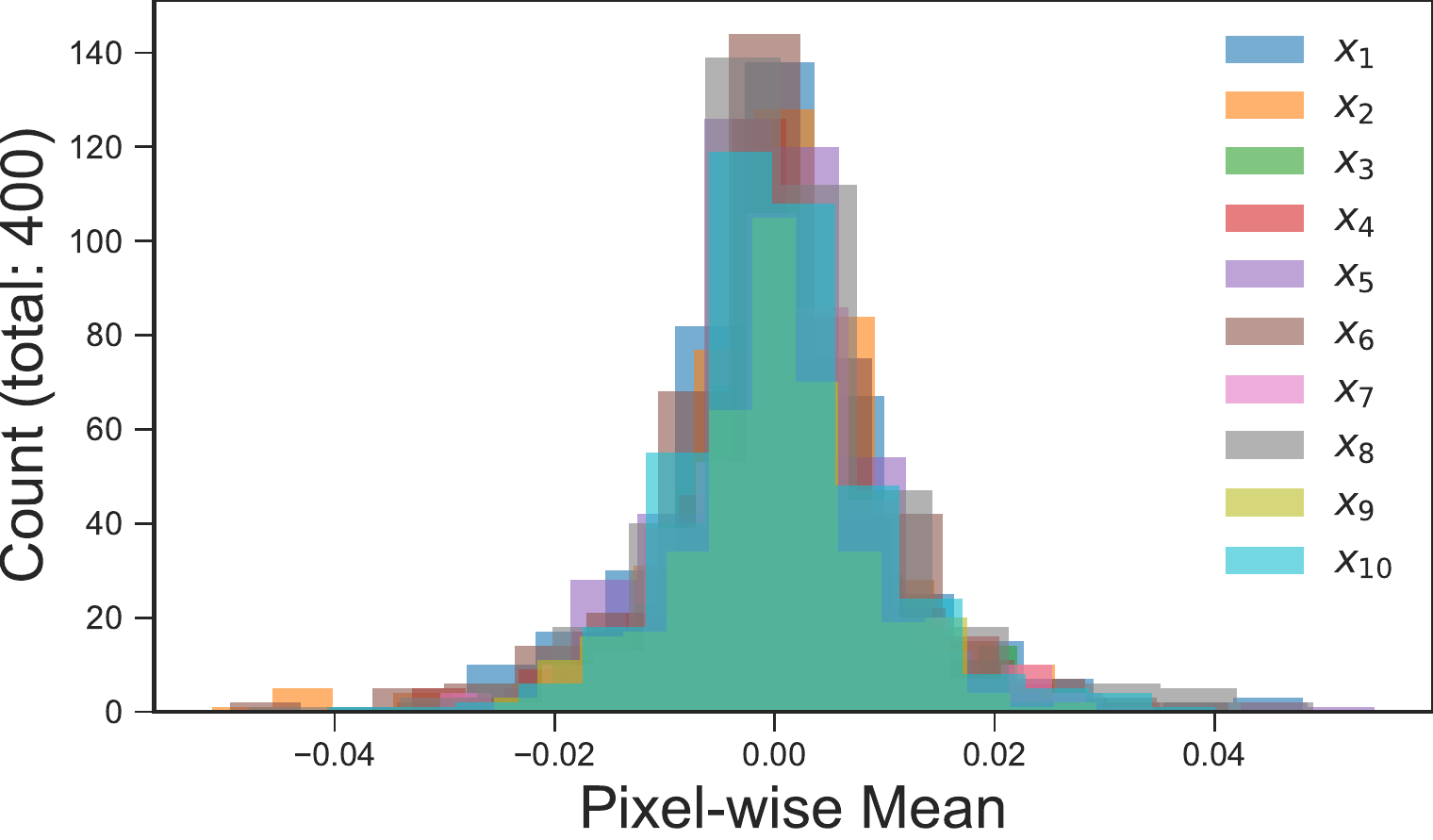}\hspace{3mm}
    \includegraphics[width=0.29\textwidth]{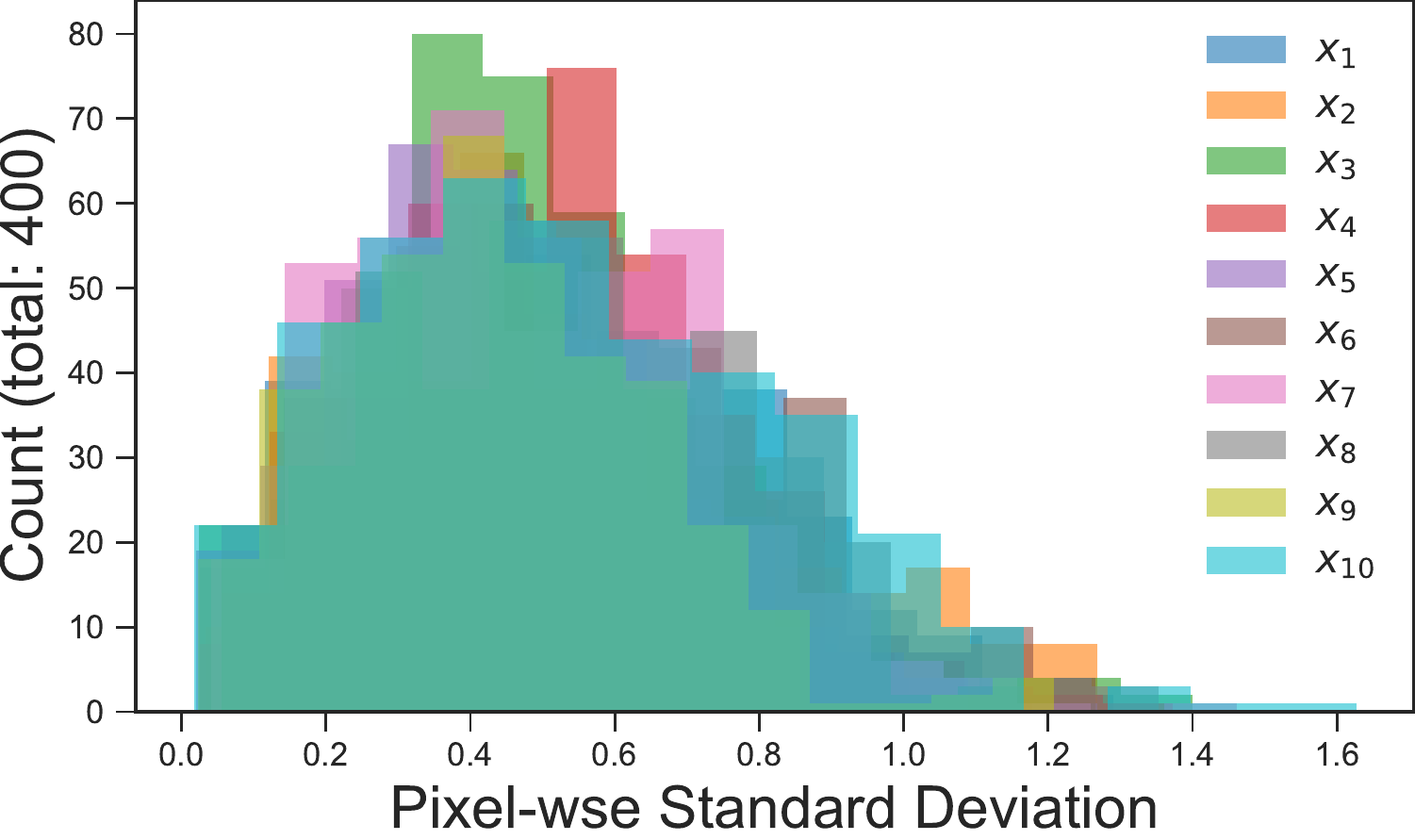}\hspace{3mm}
    \includegraphics[width=0.29\textwidth]{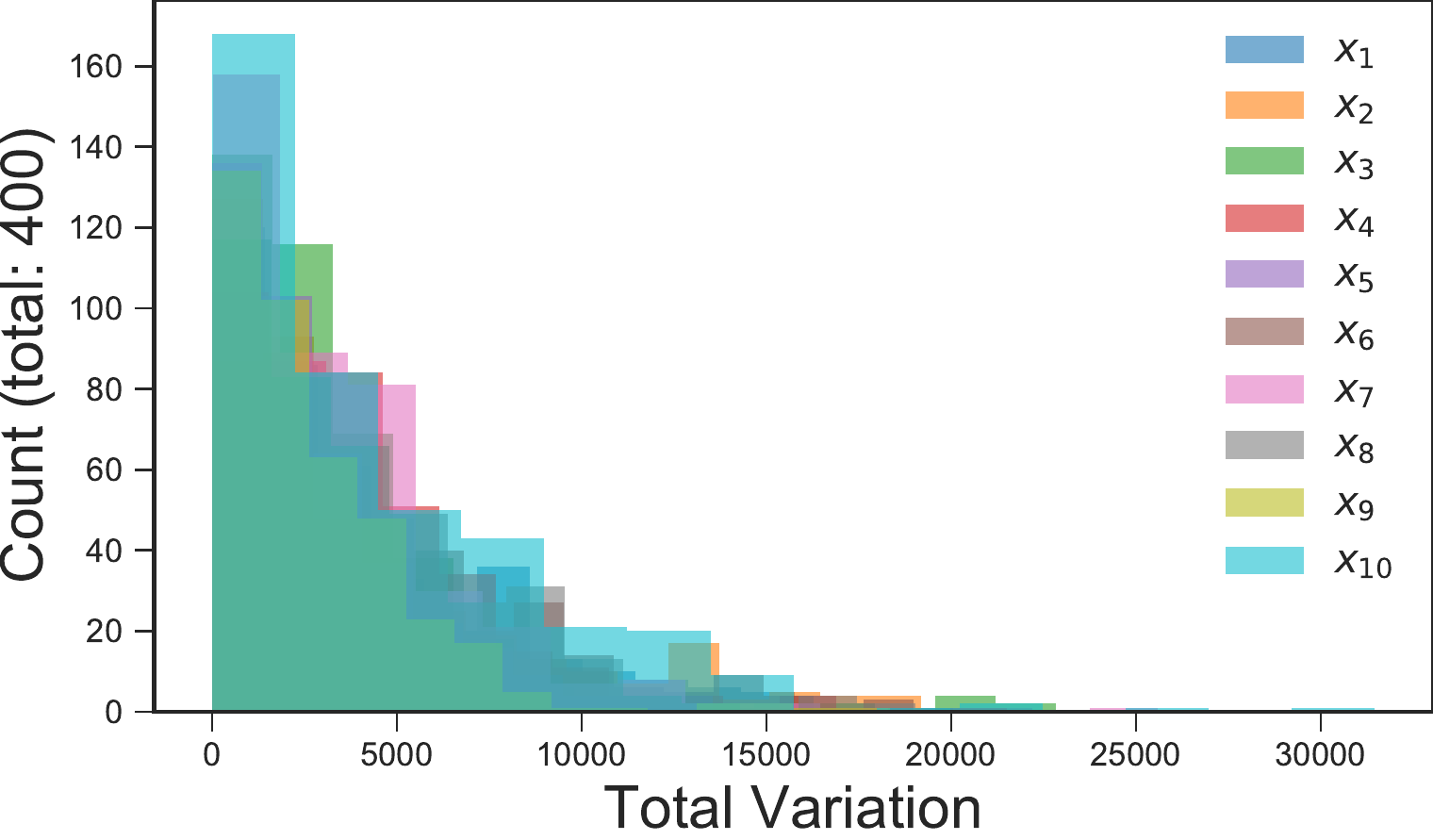}}\\
    \subfloat{
    \includegraphics[width=0.22\textwidth]{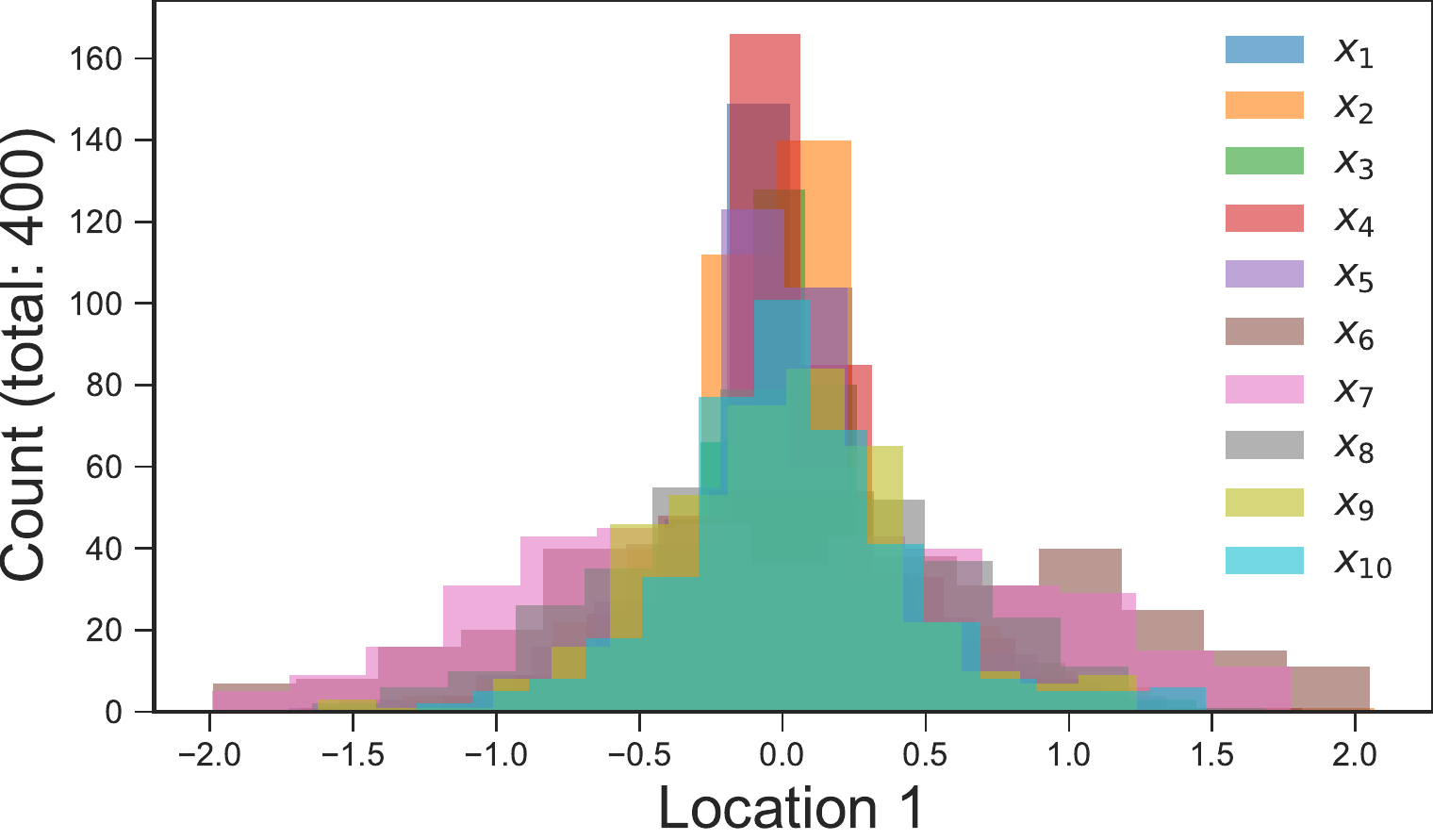}\hspace{1.8mm}
    \includegraphics[width=0.22\textwidth]{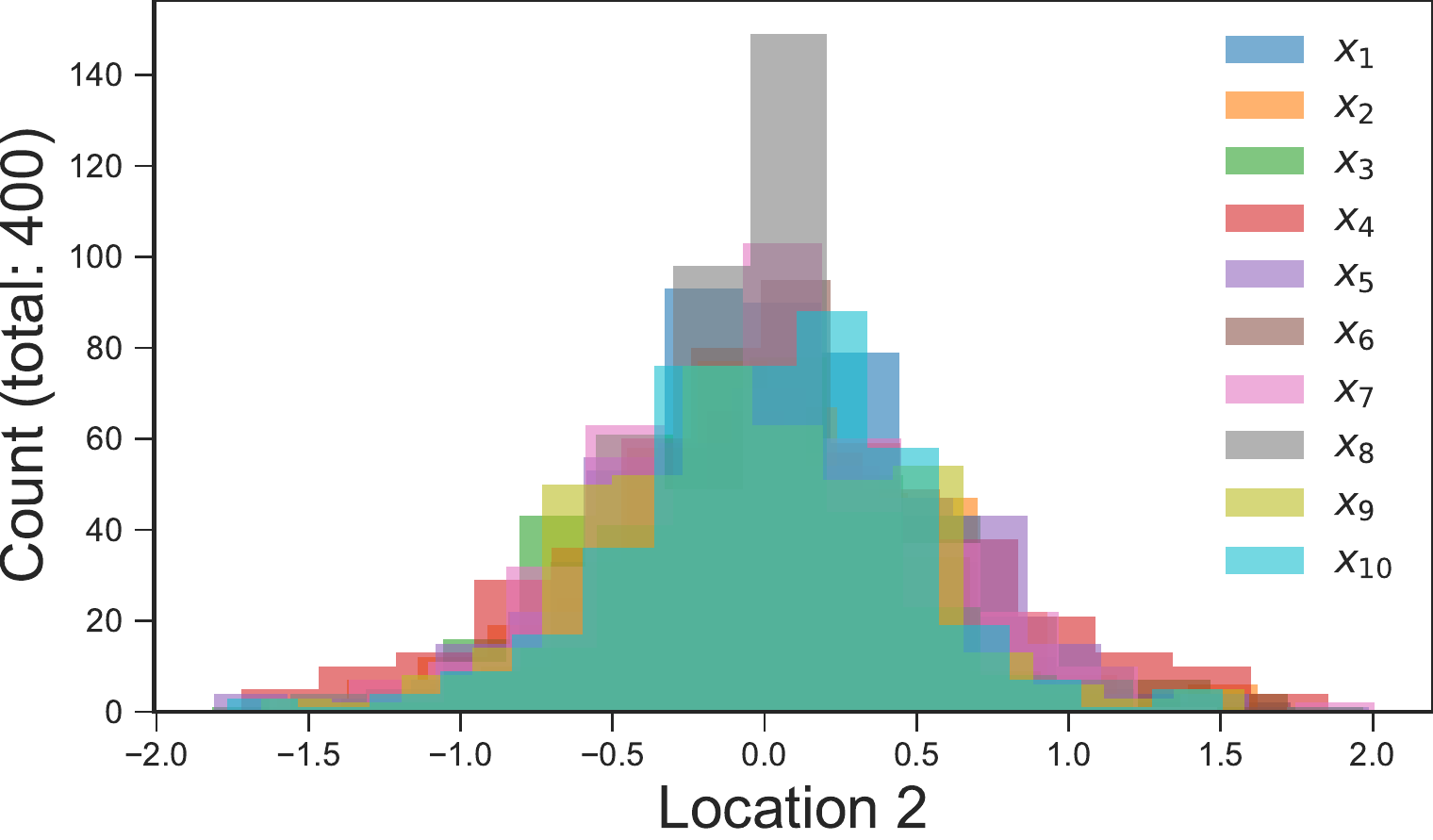}\hspace{1.8mm}
    \includegraphics[width=0.22\textwidth]{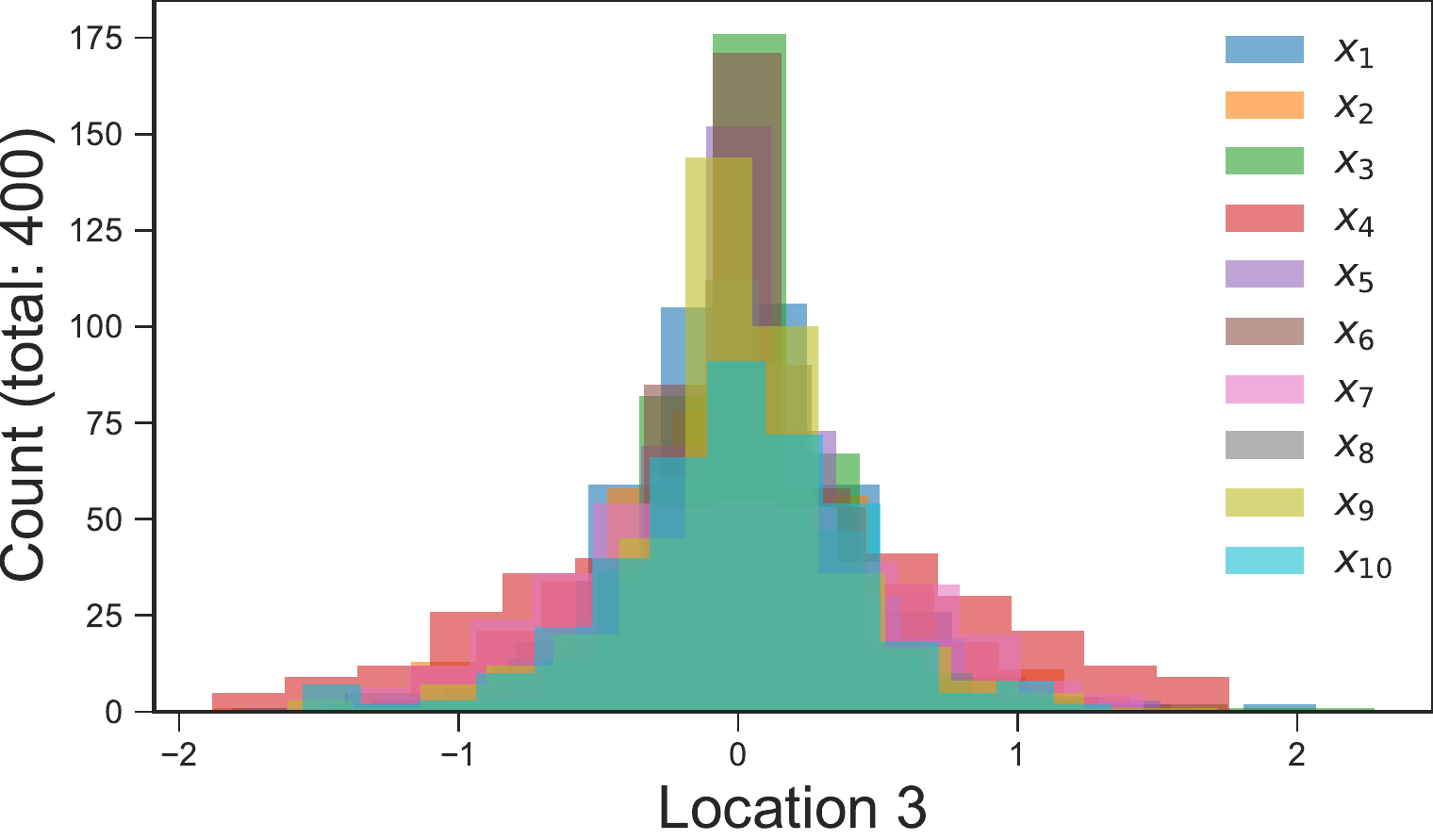}\hspace{1.8mm}
    \includegraphics[width=0.22\textwidth]{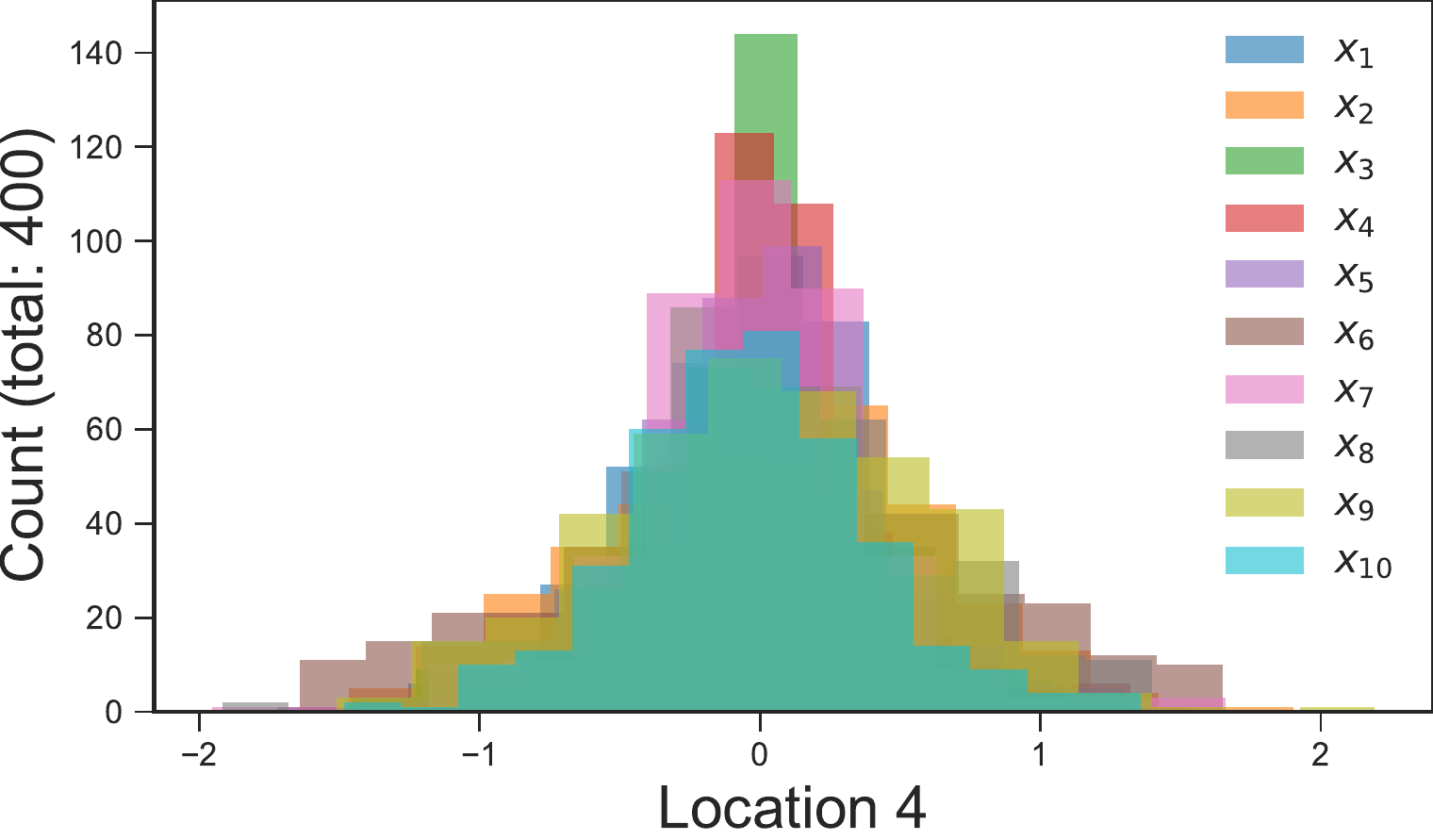}
    }
    \caption{Empirical distributions of 3 statistics (first row) and 4 random locations (second row) for 400 {\instahide} encryptions of 10 different images (cross-dataset, $k=4$). Distributions of encryptions of different images are indistinguishable. }
    \label{fig:encrypt_dist}
\end{figure*}

\subsection{Secure Encryption Implies Secure Protocol}
\label{sec:security_part1}

Suppose an attacker exists that compromises an image $x$ in the protocol. We do the thought experiment of even providing the attacker with encryptions of {\em all} images belonging to all parties, as well as model parameters in each iteration. Now everything the attacker sees during the protocol it can efficiently compute by itself, and we can convert the attacker to one that, given $\tilde{x}$, extracts information about $x$. We conclude in this thought experiment that a successful attack on the protocol also yields a successful attack on the encryption. In other words, privacy loss during protocol is {\em upper bounded} by privacy loss due to the encryption itself. Of course, this proof allows the possibility that the protocol ---due to aggregation of gradients, etc.---ensures 
even greater privacy than the encryption alone.

\paragraph{Dealing with multiple encryptions of same $x$:}
In our protocol each image $x$ is re-encrypted in each epoch. This seems to act as a data augmentation and  improves accuracy. The above argument translated to this setting shows that privacy violation requires solving the following problem: private images $x_1, x_2, \ldots, x_n$ were each encrypted $T$ times ($n$ is size of the private training set, $T$ is number of epochs), each  time using a new private key. Attacker is given these $nT$ encryptions. {\em Weaker task:} Attacker has to identify which of them came from $x_1$. {\em Stronger task:} Attacker has to identify $x_1$. 

We conjecture both tasks are hard.  Visualization (see Figure~\ref{fig:encrypt_dist}) as well as the Kolmogorov–Smirnov test~\cite{k33,s48} (see Appendix~\ref{sec:exp_app}) suggest that statistically, it is difficult to distinguish among distributions of encryptions of different images. 
Thus, effectively identifying multiple $\tilde{x}$'s of same $x$ and using them to run attacks seems difficult.

\subsection{Hardness of Attacking {\instahide} Encryption}
\label{sec:security_part2}

Now we consider the difficulty of recovering information about $x$ given a single encryption $\tilde{x}$.

\paragraph{Security estimates of naive attack.}
We start by considering the naive attack on cross-dataset {\instahide}, which would involve the attacker to either figure out the set of all $k-2$ public images, or to compromise the mask $\sigma$ and run attacks on {\mixup}.
This should take $\min \{ |\mathcal{D}_{\rm public}|^{k-2}, 2^d \}$ time. 
For cross-dataset {\instahide} schemes with a large public dataset (e.g. ImageNet), the computation cost of attack will be $10^{7(k-2)}$, and $k=4$ already makes the attack hard.

Now we suggest reasons why the naive attack may be best possible.

\paragraph{For worst-case pixel-vectors, finding the $k$-image set is hard.} Appendix~\ref{sec:hard_app} shows that for {\em worst-case} choices of images (i.e., when an \textquotedblleft image\textquotedblright\ is allowed to be an arbitrary sequence of pixel values) the computational complexity of this problem is related to the famous $k$-{\sc vector subset sum} problem.\footnote{Given a set of $N$ public vectors $v_1, \cdots v_N \in \R^d$ and $\sum_j^k v_{i_j}$, the sum of a secret subset $i_1, \cdots, i_k$ of size $k$, $k$-{\sc vector subset sum} aims to find $i_1, \cdots, i_k$. However, in our case (cross-dataset {\instahide}), only $k-2$ vectors are drawn from the public set while the other 2 are drawn from a private set. All the vectors are public in the classical setting, therefore our setting is even harder.} The attacker would need to exhaustively search over all possible combinations of public images, which is at the cost of $|\mathcal{D}_{\rm public}|^{k-2}$. Thus when $\mathcal{D}_{\rm public}$ is a large public dataset and $k \geq 4$, the computational effort to recover the image ought to be of the order of $10^{10}$ or more. 

Of course, images are not worst-case vectors. Thus an attack must leverage this fact somehow. The obvious idea today is to use a deep net for the attack, and the experiments below will suggest the obvious ideas do not work. 

\paragraph{Compromising the mask is also hard.} As previously discussed, the pixel-wise mask $\sigma \sim \Lambda_{\pm}^d$ ( Def.~\ref{def:lam_pm}) in {\instahide} is kept private by each client, which is analogous to a private key (assuming the client never shares its own $\sigma$ with others, and the generation of $\sigma$ is statistically random). Brute-force algorithm consumes $2^d$ time to figure out $\sigma$, where $d$ can be several thousands in vision tasks.

%% file: exp.tex
\section{Experiments}\label{sec:exp}

We have conducted experiments to answer three questions: 
\begin{enumerate}
    \item  How much accuracy loss does {\instahide} suffer (Section~\ref{sec:exp_acc})?
    \item How is the accuracy loss of {\instahide} compared to differential privacy approaches (Section~\ref{sec:mixup_vs_noise})? 
    \item Can {\instahide} defend against known attacks (Section~\ref{sec:exp_privacy})? 
\end{enumerate}
We are particularly interested in the cases where $k\geq4$.

\begin{figure*}
    \centering
    \subfloat[MNIST]{\includegraphics[width=0.28\textwidth]{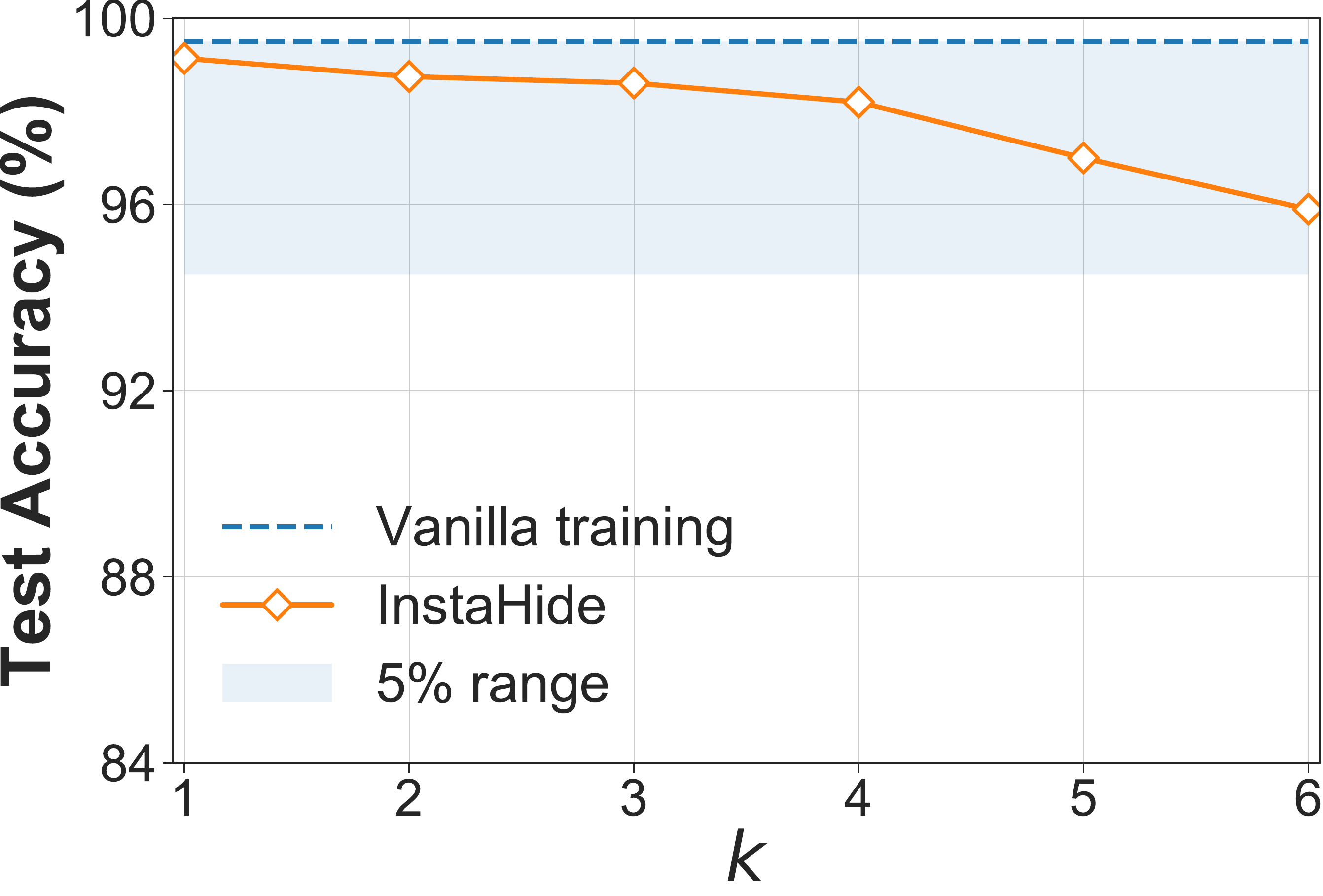}}
    \hspace{5mm}
    \subfloat[CIFAR-10]{\includegraphics[width=0.28\textwidth]{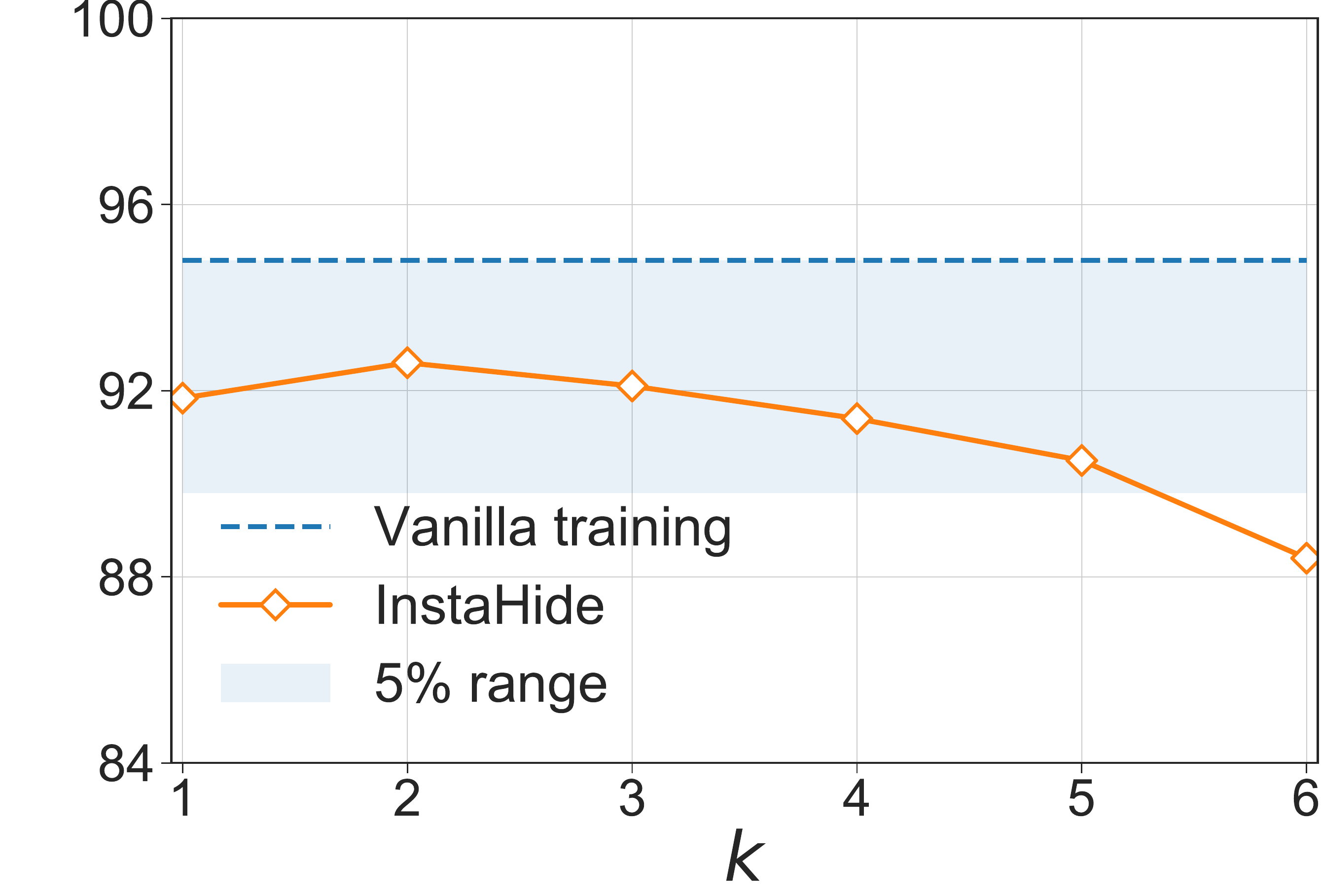}}
    \hspace{5mm}
    \subfloat[CIFAR-100]{\includegraphics[width=0.28\textwidth]{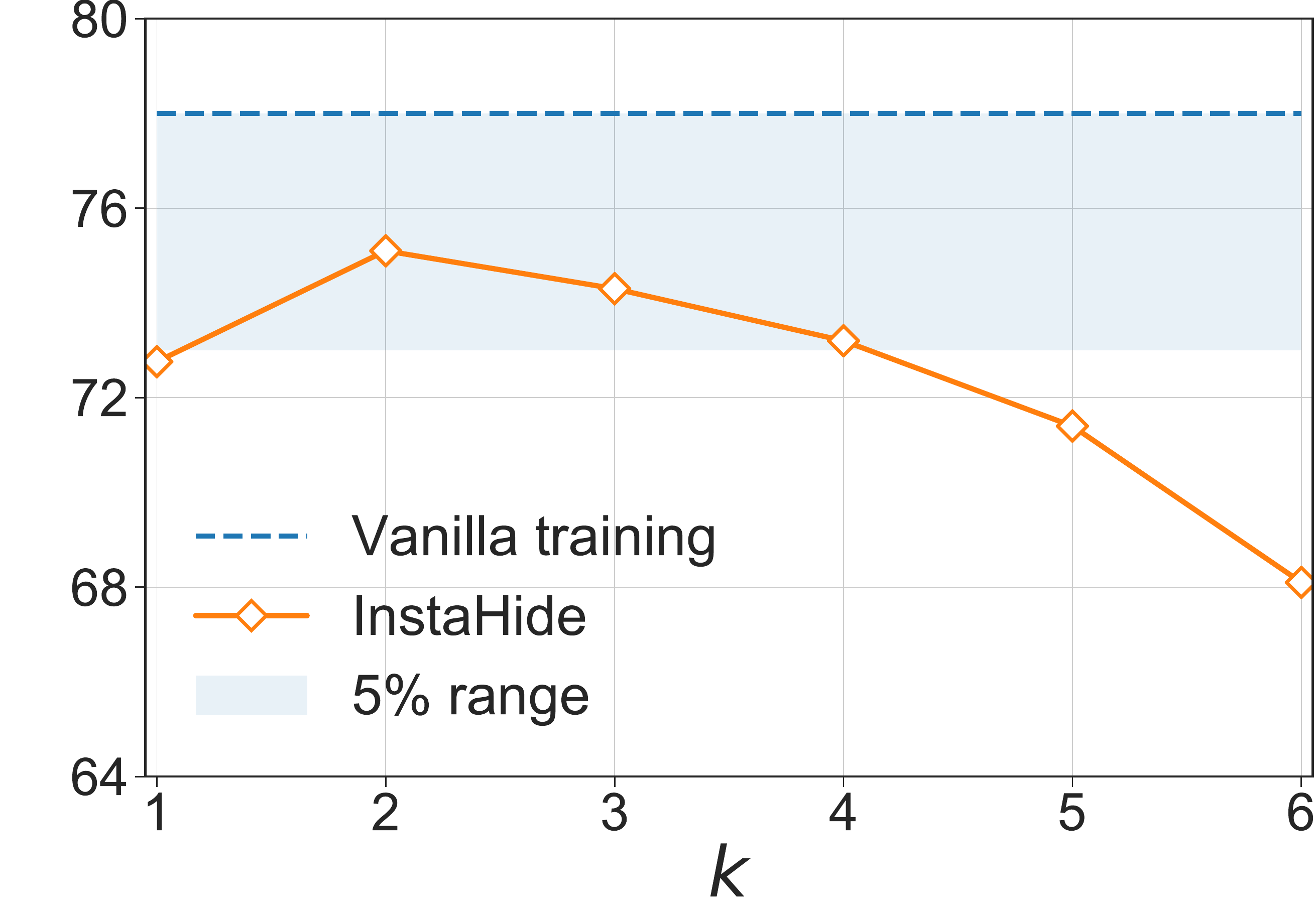}}
    \caption{Test accuracy ($\%$) on MNIST, CIFAR-10 and CIFAR-100 for vanilla training and inside-dataset {\instahide} with different $k$'s. {\instahide} with $k\leq 4$ only suffers small accuracy loss.}
    \label{fig:test_acc}
\end{figure*}

\begin{table*}[t]
\setlength{\tabcolsep}{1mm}
    \centering
    \begin{tabular}{lccccl}
    \toprule
       & {\bf MNIST}  &  {\bf CIFAR-10} & {\bf CIFAR-100} & {\bf ImageNet} & {\bf \em Assumptions}\\
       \midrule
       {\bf Vanilla training}  & $99.5 \pm 0.1$ &  $94.8 \pm 0.1$ & $77.9 \pm 0.2$ & $77.4$ & - \\
       \midrule
       {\bf DPSGD$^*$} & $98.1$ & $72.0$ & N/A & N/A &A1 \\
       \midrule
       {\bf {\instahide}$_{\text{\rm inside}, k=4, \text{ \rm in inference}}$}   & $98.2 \pm 0.2$ & $91.4 \pm 0.2$ & $ 73.2 \pm 0.2$   & $72.6$ & - \\
       {\bf {\instahide}$_{\text{\rm inside}, k=4}$} & $98.2 \pm 0.3$ &  $91.2 \pm 0.2$ &  $ 73.1 \pm 0.3$ & 1.4 &- \\
       {\bf {\instahide}$_{\text{\rm cross}, k=4, \text{ \rm in inference}}$}  &  $98.1 \pm 0.2$ & $90.3 \pm 0.2$ & $72.8 \pm 0.3$ & - & A2\\
       {\bf {\instahide}$_{\text{\rm cross}, k=4}$ } &  $97.8 \pm 0.2$ & $ 90.7 \pm 0.2$ & $ 73.2 \pm 0.2$ & - & A2 \\
       {\bf {\instahide}$_{\text{\rm cross}, k=6, \text{ \rm in inference}}$}  &   $97.4 \pm 0.2$ & $89.6 \pm 0.3$ & $72.1 \pm 0.2$ & - & A2\\
       {\bf {\instahide}$_{\text{\rm cross}, k=6}$ } & $97.3 \pm 0.1$  & $89.8 \pm 0.3$ & $71.9 \pm 0.3$  & - & A2\\
     \bottomrule
    \end{tabular}
    \caption{Test accuracy ($\%$) on MNIST, CIFAR-10, CIFAR-100 and ImageNet for vanilla training, DPSGD~\cite{acg+16} and {\instahide}, including the mean and standard deviation of test accuracy across 5 runs except for ImageNet. 
    Results marked with ``in inference''  applies {\instahide} during inference.  {\instahide} methods incur minor accuracy reductions. A1 denotes ``{\em Using a {\bf \em labelled} public dataset for pre-training}'', A2 denotes `` {\em Using a large {\bf \em unlabelled} public dataset for hiding''}.
    $^*$DPSGD results are from~\cite{dpsgd19}, which does not have results for CIFAR-100 and ImageNet.
    } 
    \label{tab:test_acc}
\end{table*}

\begin{figure*}[t]
\centering
\subfloat[Accuracy with different $\alpha$'s. ]{\includegraphics[width=0.26\textwidth]{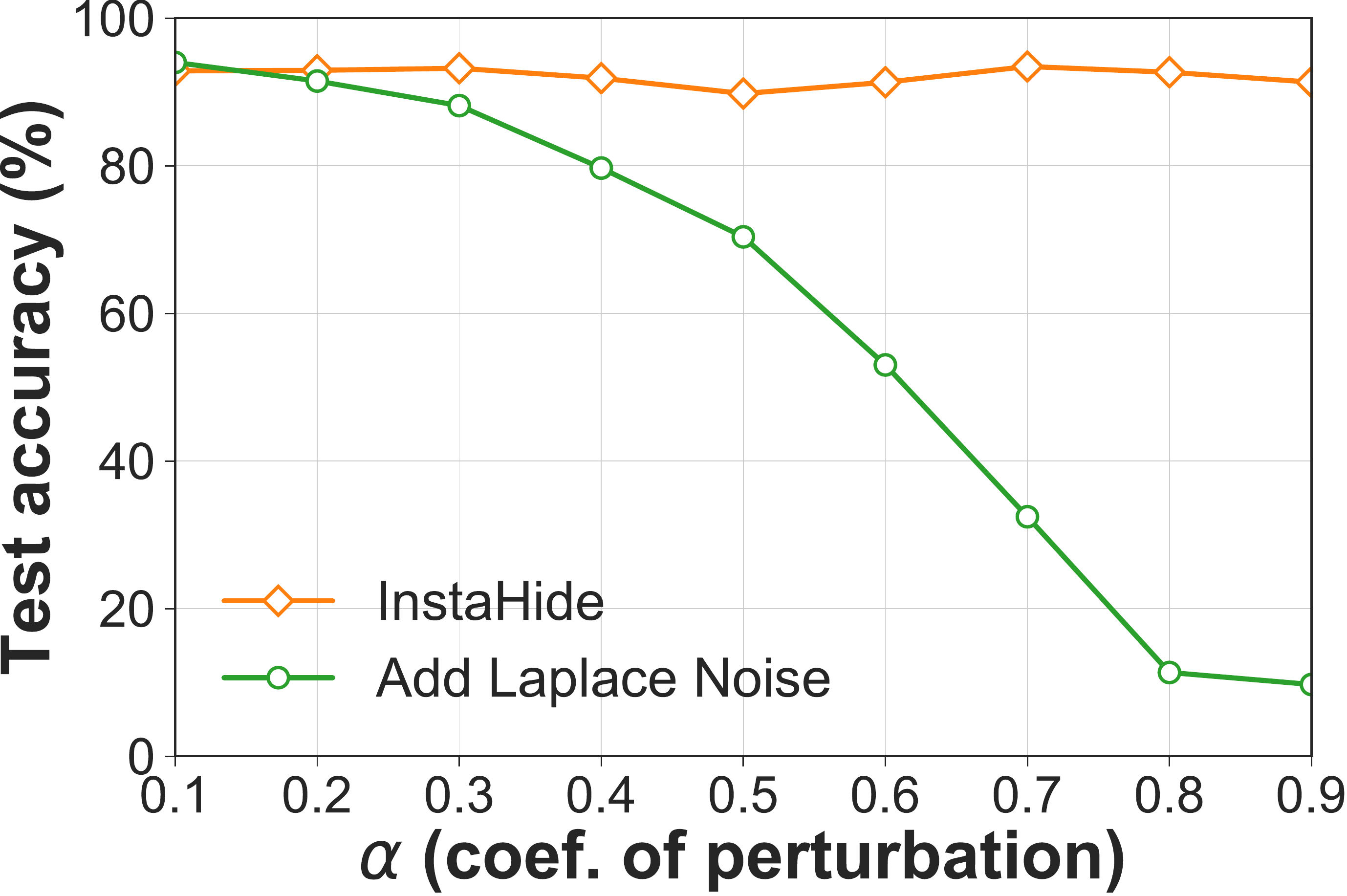}} \hspace{2mm}
\subfloat[Images generated by {\instahide} and adding Laplace noise with different $\alpha$'s. ]{\includegraphics[width=0.72\textwidth]{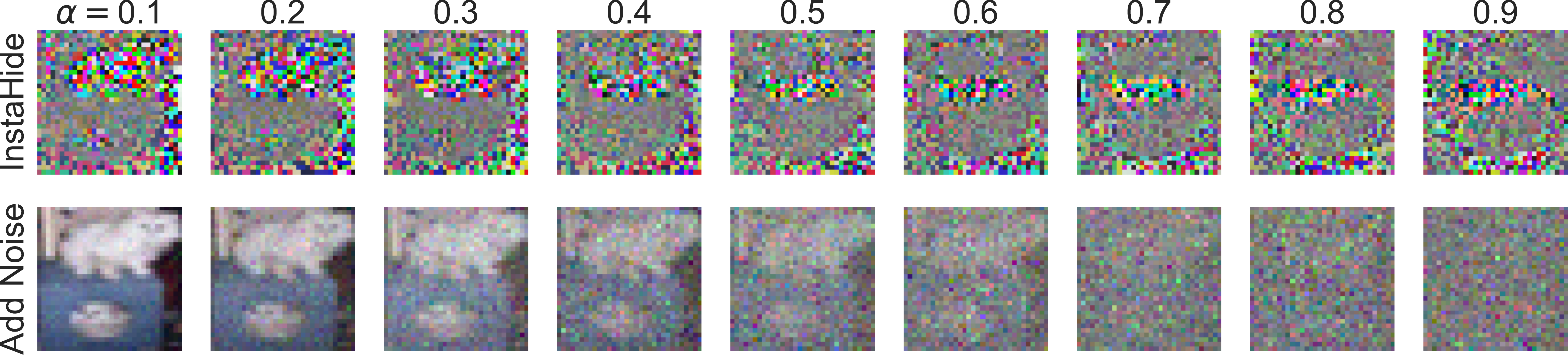}}
\caption{Relative accuracy on CIFAR-10 (a) and visualization (b) of  {\instahide} and adding random Laplace noise with different $\alpha$'s, the coefficient of perturbation.  {\instahide} gives better test accuracy than adding random noise.}
\label{fig:mixup_vs_noise_acc}
\end{figure*}

\paragraph{Datasets and setup.}
Our main experiments are image classification tasks on four datasets MNIST \cite{mnist10}, CIFAR-10, CIFAR-100~\cite{cifar10}, and ImageNet~\cite{imagenet09}.
We use ResNet-18~\cite{hzrs16} architecture for MNIST and CIFAR-10, NasNet~\cite{zoph2018learning} for CIFAR-100, and ResNeXt-50~\cite{Xie2016} for ImageNet. The implementation uses the Pytorch \cite{pytorch19} framework. Note that we convert greyscale MNIST images to be 3-channel RGB images in our experiments. Hyper-parameters are provided in Appendix~\ref{sec:exp_app}.

\subsection{Accuracy Results of {\instahide}}\label{sec:exp_acc}

We evaluate the following {\instahide} variants (with $c_1 = 0.65, c_2 = 0.3$):

\begin{itemize}
    \item Inside-dataset {\instahide} with different $k$'s, where $k$ is chosen from $\{1,2,3,4,5,6\}$. 

    \item Cross-dataset {\instahide} with $k=4, 6$. For MNIST, we use CIFAR-10 as the public dataset; for CIFAR-10 and CIFAR-100, we use the preprocessed ImageNet as the public dataset. We do not test Cross-dataset {\instahide} on ImageNet since its sample size is already large enough for good security.
\end{itemize}

The computation overhead of {\instahide} in terms of extra training time in our experiments is smaller than $5\%$.

\paragraph{Accuracy with different $k$'s.} Figure~\ref{fig:test_acc} shows the test accuracy of vanilla training, and inside-dataset {\instahide} with different $k$'s on MNIST, CIFAR-10 and CIFAR-100 benchmarks. Compared with vanilla training, {\instahide} with $k=4$ only suffers small accuracy loss of $ 1.3\%$, $3.4\%$, and $4.7\%$ respectively. Also, increasing $k$ from 1 (i.e., apply mask on original images, no {\mixup}) to 2 (i.e., apply mask on pairwise mixed images) improves the test accuracy for CIFAR datasets, suggesting that {\mixup} data augmentation also helps for encrypted {\instahide} data points.

\paragraph{Inside-dataset v.s. Cross-dataset.} We also evaluate the performance of Cross-dataset {\instahide}, which does encryption using random images from both the private dataset and a large public dataset.  As shown in Table~\ref{tab:test_acc}, Cross-dataset {\instahide} incurs an additional small accuracy loss comparing with inside-dataset {\instahide}. The total accuracy losses for MNIST, CIFAR-10 and CIFAR-100 are $1.7\%$, $4.1\%$, and $4.7\%$ respectively.
As previously suggested, a large public dataset provides a stronger notion of security.

\paragraph{Inference with and without {\instahide}.}

As mentioned in Sec~\ref{sec:inference}, by default, 
{\instahide} is applied during inference.  In our experiments, we average predictions of 10 encryptions of a test image. We found that for high-resolution images,
applying {\instahide} during inference is important: the results of using Inside-dataset {\instahide} on ImageNet in Table~\ref{tab:test_acc} show that the accuracy of inference with {\instahide} is $72.6\%$, whereas that without {\instahide} is only $1.4\%$.

\subsection{{\instahide}  vs. Differential Privacy approaches}
\label{sec:mixup_vs_noise}

Although {\instahide} is qualitatively different from differential privacy in terms of privacy guarantee, we would like to provide hints for their relative accuracy (Question 2).

\paragraph{Comparison with DPSGD.}
DPSGD~\cite{acg+16} injects  noise to gradients to control private leakage. Table~\ref{tab:test_acc} shows that DPSGD leads to an accuracy drop of about $20\%$ on CIFAR-10 dataset. By contrast, {\instahide} gives models almost as good as vanilla training in terms of test accuracy.

\paragraph{Comparison with adding random noise to images.} 
We also compare {\instahide} (i.e., adding {\it structured} noise) with adding random noise to images (another typical approach to preserve differential privacy). 
Specifically, given the original dataset $\mathcal{X}$, and a perturbation coefficient $\alpha \in (0, 1)$, we test a) {\instahide}$_{{\rm inside}, k=2}$: $\tilde{x}_i =  \sigma \circ ((1-\alpha) x_i + \alpha x_j)$, where $x_j \in \mathcal{X}$, $j \neq i$, and b) adding random Laplace noise $e$: $\tilde{x}_i = (1-\alpha) x_i + e$, where $\E[\|e\|_1] = \E[\|\alpha x_j\|_1]$.  

As shown in Figure~\ref{fig:mixup_vs_noise_acc}, by increasing $\alpha$ from 0.1 to 0.9, the test accuracy of adding random noise drops from $\sim94\%$ to  $\sim10\%$, while the accuracy of {\instahide} is above $90\%$.

\subsection{Robustness Against Attacks}
\label{sec:exp_privacy}

To answer the question how well {\instahide} can defend known attacks, here we report our findings with a sequence of attacks on {\instahide} encryption $\tilde{x}$ to recover original image $x$, including
gradient-matching attack~\cite{zlh19}, demasking using GAN (Generative Adversarial Network), averaging multiple encryptions, and uncovering public images with similarity search.

\paragraph{Gradient-matching attack \cite{zlh19}.}
\label{sec:exp_attack_gradmatch}

Here attacker observes gradients of the loss generated using a user's private image $s$ while training a deep net (attacker knows the deep net, e.g., as a participant in Federated Learning) and tries to recover $s$ by computing an image $s^*$ that has similar gradients to those of $s$ (see algorithm in Appendix \ref{sec:exp_app}). 
Figure~\ref{fig:mix_cifar} shows results of this attack on {\mixup} and {\instahide} schemes on CIFAR-10. If {\mixup} with $k=4$ is used, the attacker can still extract fair bit of information about the original image. However, if {\instahide} is used the attack isn't successful.

\paragraph{Demask using GAN.}  {\instahide} does pixel-wise random sign-flip after applying {\mixup} (with public images, in the most secure version). This flips the signs of half the pixels in the mixed image. An alternative way to think about it is that the  adversary sees the intensity information (i.e. absolute value) but not the sign of the pixel. Attackers could use computer vision ideas to  recover the sign. One attack consists of training a GAN on this sign-recovery task\footnote{We thank Florian Tramèr for suggesting this attack.}, using a large training set of  $(z, \sigma \circ z)$ where $z$ is a mixed image and $\sigma$ is a random mask. If this GAN recovers the signs reliably, this effectively removes the mask, after which one could use the attacks against {\mixup} described in Section~\ref{sec:warmup}.

In experiments this only succeeded in recovering half the flipped signs, which means $\sim 1/4$ of the coordinates continued to have the wrong sign; see Figure~\ref{fig:GAN_demask}. GAN training\footnote{We use this GAN architecture (designed for image colorization): \href{https://github.com/zeruniverse/neural-colorization}{https://github.com/zeruniverse/neural-colorization}.} used 50,000 cross-dataset {\instahide} examples generated with CIFAR-10 and ImageNet ($k=6$). This level of sign recovery seems insufficient to allow the attack against {\mixup} (Section~\ref{sec:warmup}) to succeed, nor the other attacks discussed below. Nevertheless, researchers trying to break {\instahide} may want to use such a demasking GAN as a starting point.

\begin{figure}[t]
  \centering
  \subfloat[Original]{\includegraphics[width=0.18\linewidth]{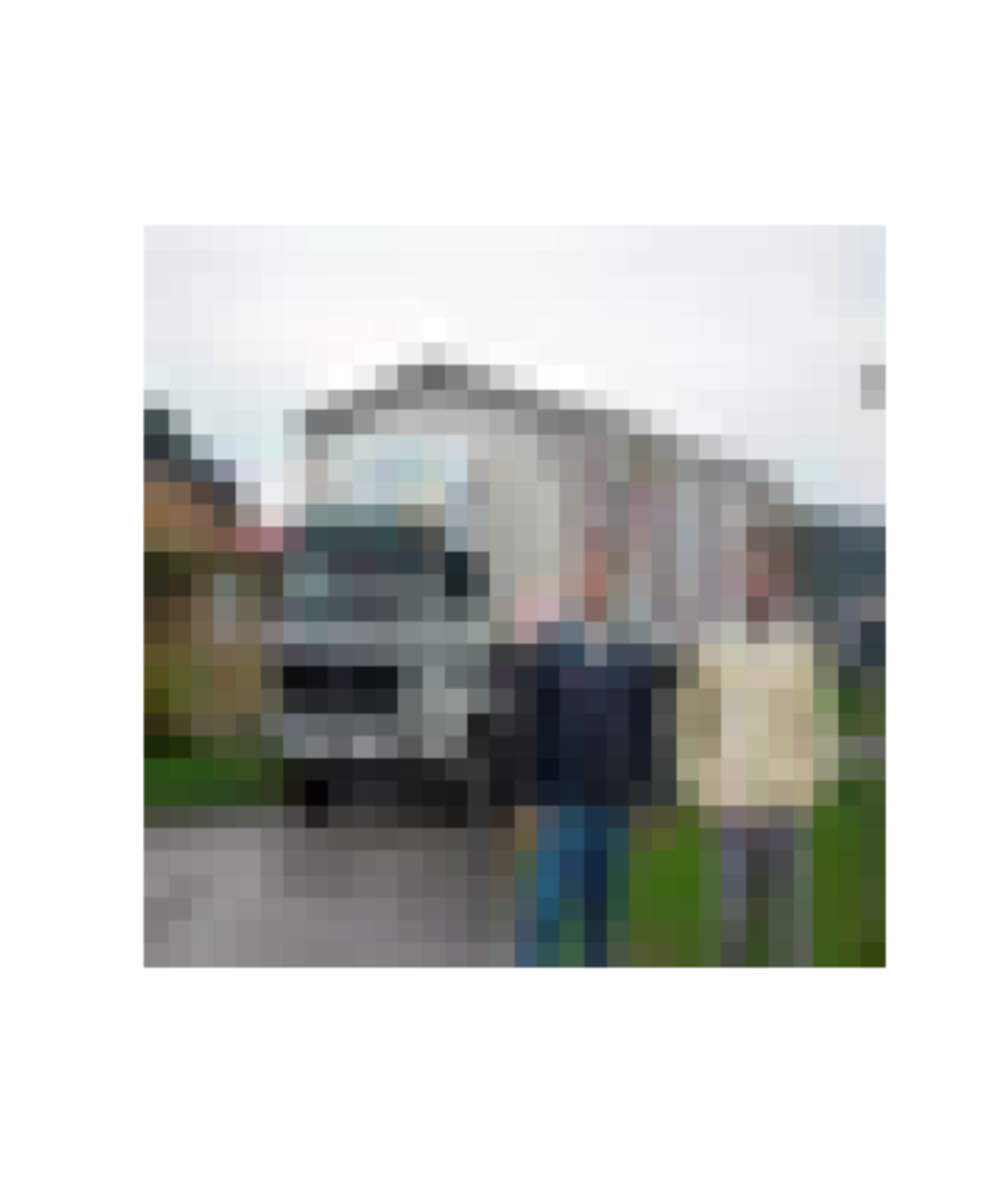}}\hspace{1.5mm}
  \subfloat[ After mixing]{\includegraphics[width=0.365\linewidth]{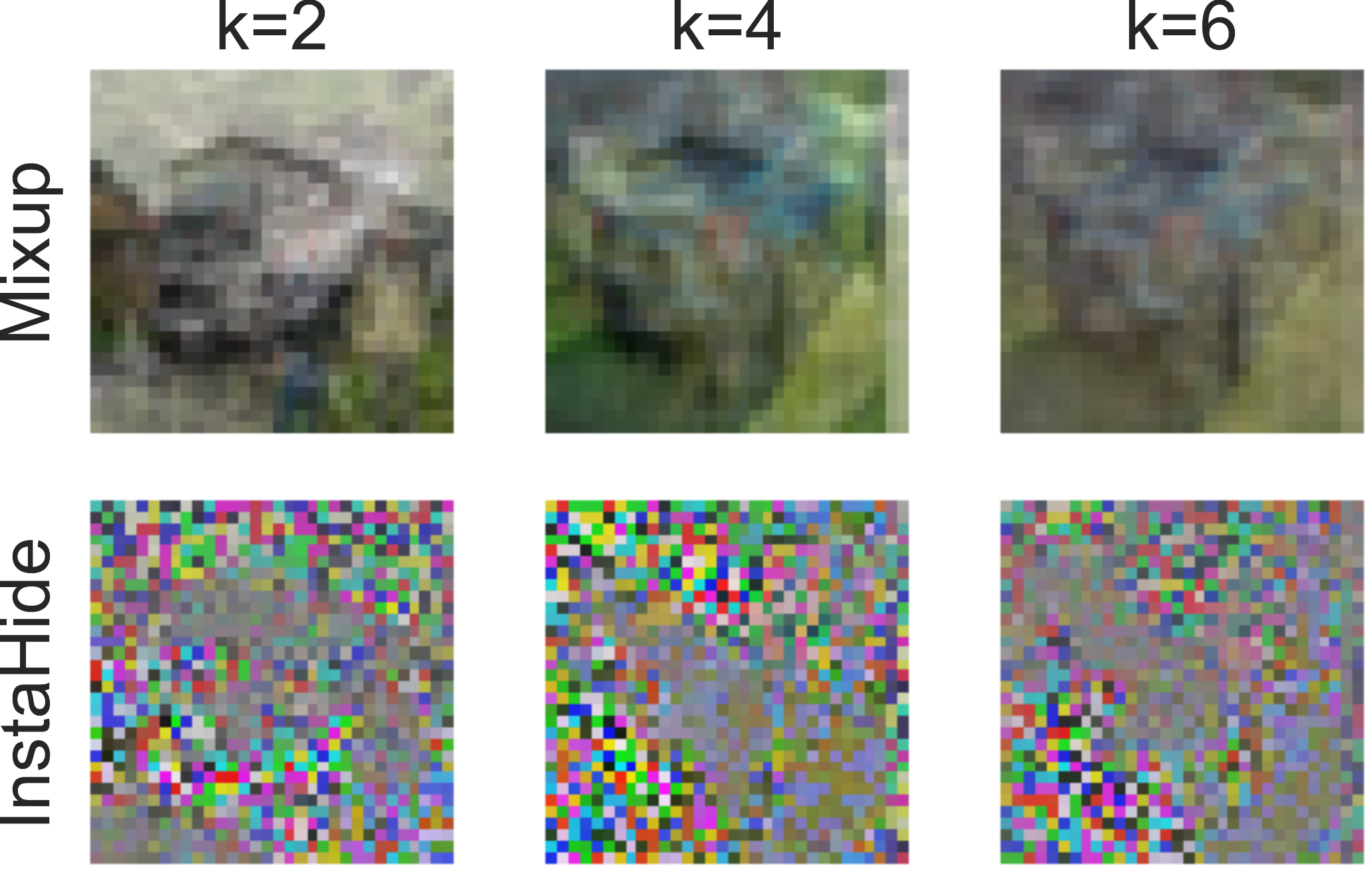}}
  \hspace{1.5mm}
  \subfloat[ Recovered by attack]{\includegraphics[width=0.365\linewidth]{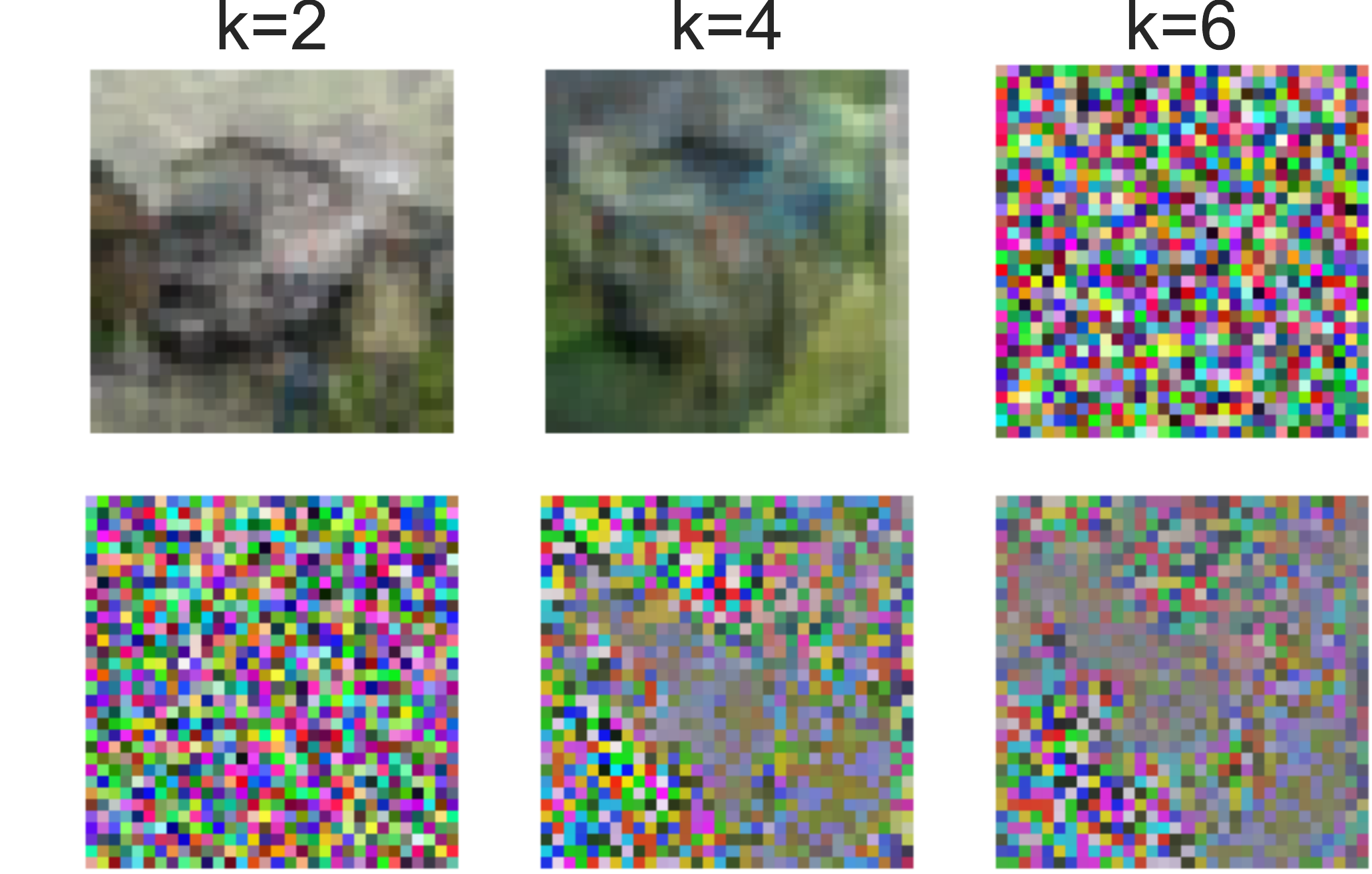}}
  \caption{Visualization of (a) the original image, (b) {\mixup} and {\instahide} images, and (c) images recovered  by the {\em gradients matching attack}. {\instahide} is more effective in hiding the image than {\mixup}.}\label{fig:mix_cifar}
\end{figure}

\begin{figure}[!t]
    \centering
    \includegraphics[width=0.8\linewidth]{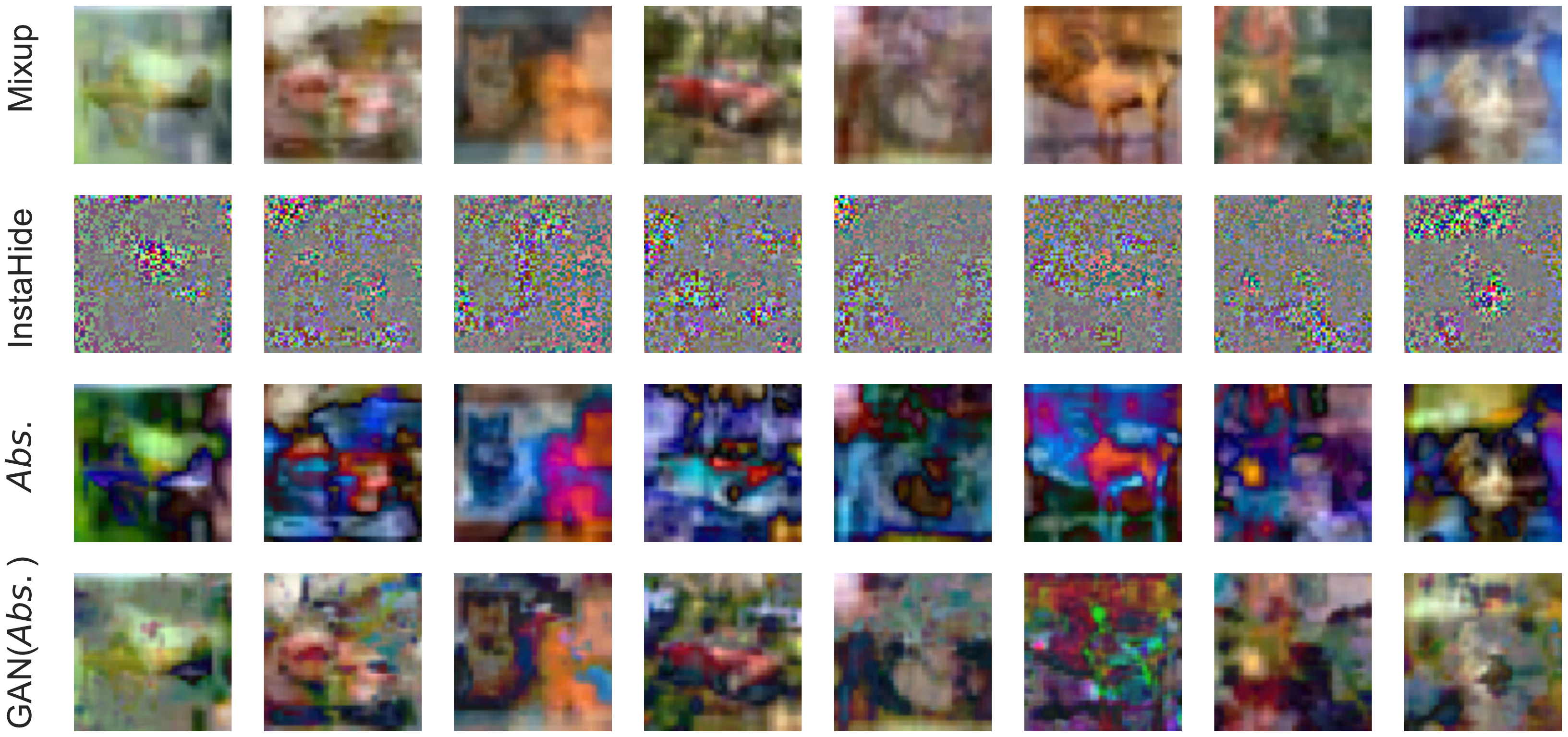}
    \caption{Undo sign-flipping using GAN. Rows: (1) Output  from {\mixup} algorithm (Algorithm~\ref{alg:mixup}, line~\ref{line:mix_x}); (2) Result of applying random mask $\sigma$ on previous row (Algorithm~\ref{alg:instahide}, line \ref{lin:compute_wt_x_i}). Note that randomly flipping sign of a pixel still preserves its absolute value. 
    (3) Taking coordinate-wise absolute value of previous row. (4) Output of demasking GAN on previous row. The attack corrects about $1/4$ of the flipped signs but this doesn't appear enough to allow further attacks that recover the encrypted image.
    }
    \label{fig:GAN_demask}
\end{figure}

\begin{figure}[!t]
    \centering
    \includegraphics[width=0.8\linewidth]{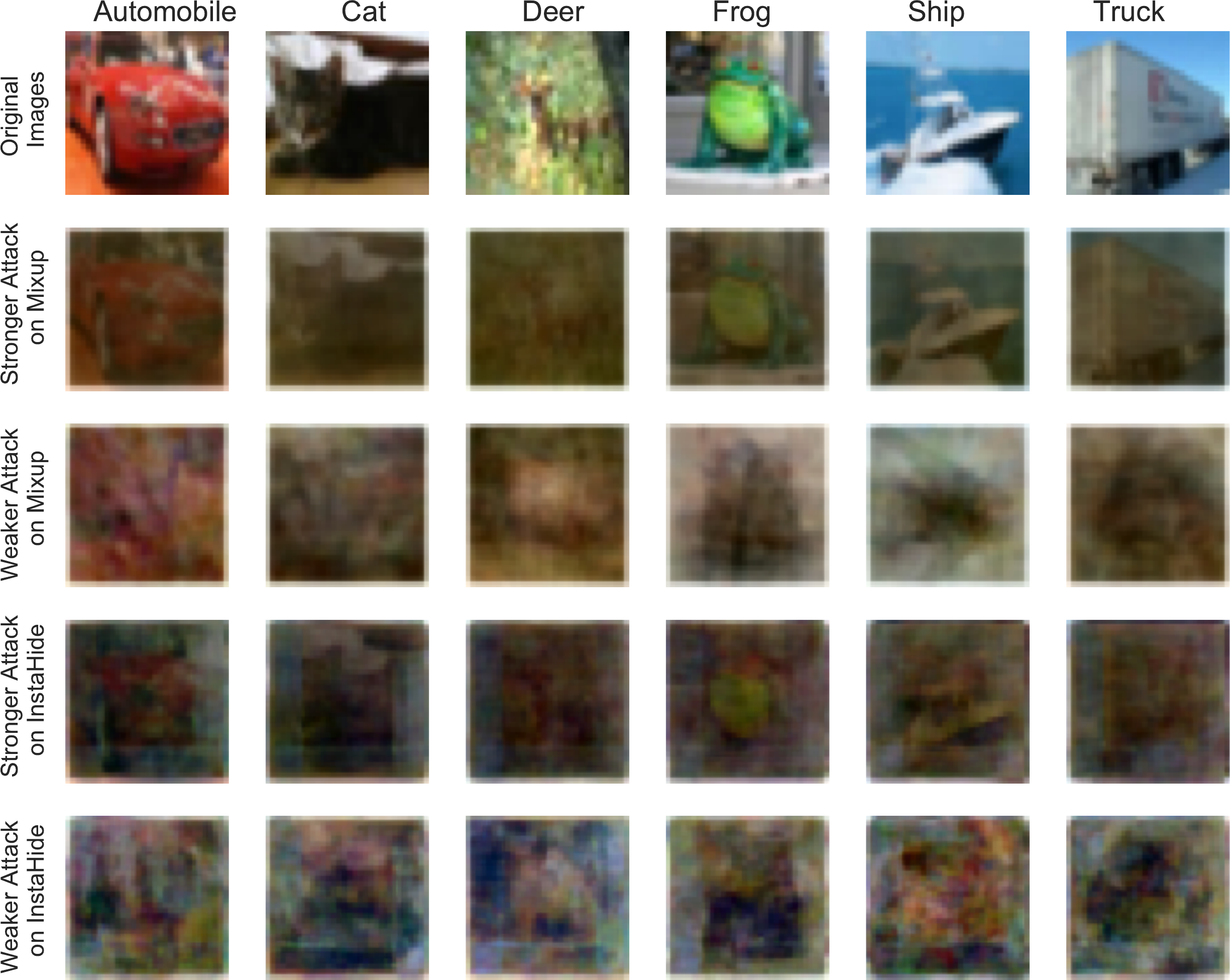}
    \caption{Average multiple encryptions (after demasking) of the same image to attack ($k=6$). In the stronger attack, the attacker already knows the set of multiple encryptions of the same image; in the weaker attack, the attacker has to identify that set first. Rows: (1) Original images. (2-3) Results of the stronger and the weaker attacks on {\mixup}. (4-5) Results of the stronger and the weaker attacks on {\instahide}. Note that the attacker has to demask {\instahide} encryptions using GAN before running attacks, and the information loss of this step makes it harder to attack {\instahide} than the plain {\mixup}.
    } 
    \label{fig:average_attack}
\end{figure}

\paragraph{Average multiple encryptions of the same image.}

We further test if different encryptions of the same image (after demasking) can be used to recover that hidden image by running the attack in  Section~\ref{sec:mixup_inside_attack}.

Assuming a public history of $n \times T$ encryptions, where $n$ is the size of the private training set, and $T$ is the number of epochs.
We consider a stronger and a weaker version of this attack. 
\begin{itemize}
    \item {\bf Stronger attack}: the attacker already {\em knows} the set of multiple encryptions of the same image $x$. He uses GAN to demask all encryptions in the set, and averages images in the demasked set to estimate $x$. 
    \item {\bf Weaker attack}: the attacker {\em does not know} which subset of the encryption history correspond to the same original image. To identify that subset, he firstly demasks all $n\times T$ encrytions in the history using GAN. With an arbitrary demasked encryption (from the history) for some unknown original image $x$, he runs similarity search to find top-$m$ closest images in $n\times T -1$ other demasked encryptions (which may also contain $x$), and averages these $m+1$ images to estimate $x$.
\end{itemize}

The stronger attack is conceivable if $n$ is very small (say a hospital only has 100 images), so via brute force the attacker can effectively have a small set of encryptions of the same image. However, in practice, $n$ is usually at least a few thousand.

For simplicity, we test with $n=50$ and $T=50$ (a larger $n$ will make the attack harder). We use the structural similarity index measure (SSIM)~\cite{wang2004image} as the similarity metric, and set $m$ to 5 after tuning. 

We also run this attack directly on {\mixup} for comparison. As shown in Figure~\ref{fig:average_attack}, if the original image is not flat (e.g. the ``deer''), the stronger attack may not work. For flat images (e.g. the ``truck'') or images with strong contrast (e.g. the ``automobile" and the ``frog''), the stronger attack is able to vaguely recover the original image. However, as previously suggested, the stronger attack is feasible only for a very small $n$.

Note that results here is an upper bound on privacy leakage  since we assume a perfect recovery of $\tilde{x}$ from the gradients. In real-life scenarios  this may not hold.

\begin{figure}[t]
    \centering
        \includegraphics[width=0.6\linewidth]{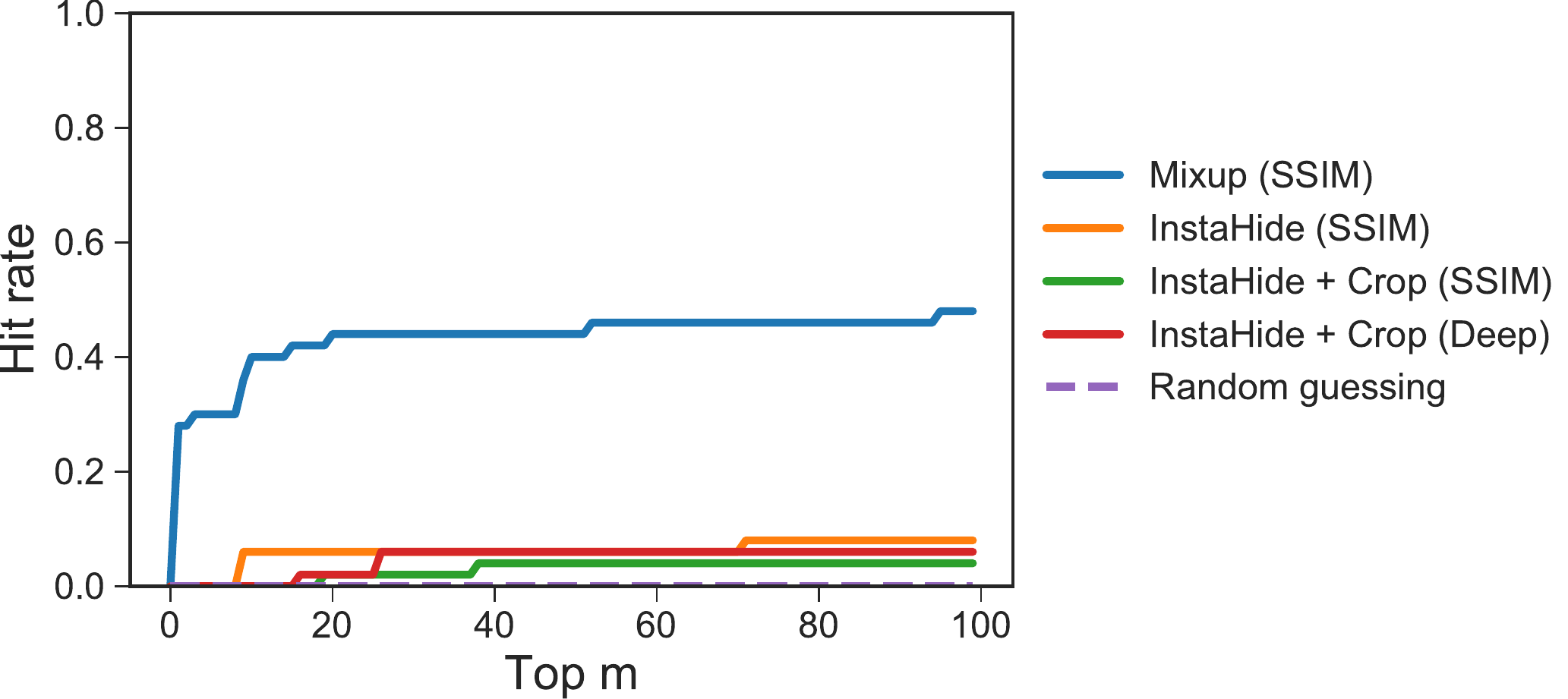}
    \caption{Averaged hit rate of uncovering public images for mixing among the top-$m$ answers returned by similarity search. Running this attack on {\instahide} requires GAN-demasking as the first step, and a wrongly demasked pixel will make the similarity score less reliable and yield lower hit rate. Mixing with random cropped patches of public images augments the public dataset and gives more security. It also disables similarity search using SSIM. Train a deep model to predict similarity score also does not give a promising hit rate. }
    \label{fig:public_attack}
\end{figure}

\paragraph{Uncover public images by similarity search.} We also run the attack in Section~\ref{sec:mixup_attack_public} after demasking {\instahide} encryptions using GAN, which tries to uncover the public images for mixing by running similarity search in the public dataset using the demasked encryption as the query.

We test with $k=6$: mix 2 private images from CIFAR-10 with 4 public images from a set of 10,000 ImageNet images (i.e. $N=10,000$). We consider the attack a `hit' if at least one public image for mixing is among the top-$m$ answers of the similarity search. The attacker uses SSIM as the default similarity metric for search. However, a traditional alignment-based similarity metric (e.g SSIM) would fail in {\instahide} schemes which use randomly cropped patches of public images for mixing (see Figure~\ref{fig:public_attack}), so in that case, the attacker trains a deep model (VGG~\cite{sz15} in our experiments) to predict the similarity score.

Note that to find the 4 correct public images for mixing, the  attack has to try all ${m \choose 4}$ combinations of the top-$m$ answers with different coefficients, and subtract the combined image from the demasked {\instahide} encryption to verify. Figure~\ref{fig:public_attack} reports the averaged hit rate of this attack on 50 different {\instahide} images.  As shown, even with a relatively small public dataset ($N=10,000$) and a large $m=\sqrt{N}$, the hit rate of this attack on {\instahide} (enhanced with random cropping) is around 0.05 (i.e. the attacker still has to try ${m \choose 4} = O(N^2)$ combinations to succeed with probability 0.05). Also,  this attack appears to become much more expensive with the public dataset being the whole ImageNet dataset ($N=1.4\times 10^7$) or random images on the Internet.

\section{{\instahide} Deployment: Best practice}
\label{sec:privacy_suggestion}

Based on our security analysis (sec~\ref{sec:privacy} and sec~\ref{sec:exp_privacy}), we suggest the following:

\begin{itemize}
    \item 
    Consider Inside-dataset {\instahide} only if the private dataset is very large and images have varied, complex patterns.  If the images in the private dataset have simple signal patterns or the dataset size is relatively small, consider using Cross-dataset {\instahide}.
        
    \item For Cross-dataset {\instahide}, use a very large  public dataset. Follow the preprocessing steps advocated in Sec~\ref{sec:cross_dataset} to randomly crop patches from each image in the public dataset and filter out ``flat'' patches.  
       
    \item Re-encrypt images in each epoch. This allows the benefits of greater data augmentation for deep learning and hinders attacks (as suggested in  Sec~\ref{sec:exp_privacy}).
    
    \item Since images are re-encrypted in each epoch, for best security (e.g, against gradient-matching attacks), each participant should perform a random re-batching so that batch gradients do not correspond to the same subset of underlying images.

    \item Choose $k=4,5,6$ for a good trade-off between accuracy and security.

    \item Set a conservative upper threshold for the coefficients in mixing (e.g. $0.65$ in our experiments).
\end{itemize}

\section{A Challenge Dataset} 
\label{sec:challenge_dataset}

To encourage readers to design stronger attacks, we  release a challenge dataset\footnote{ \href{https://github.com/Hazelsuko07/InstaHide_Challenge}{https://github.com/Hazelsuko07/InstaHide\_Challenge}.} of encrypted images generated by applying Cross-dataset {\instahide} with $k=6$ on some private image dataset and a preprocessed ImageNet as the public dataset. An attack is considered to succeed if it substantially recovers a significant fraction of original images.

%% file: relat.tex
\section{Related Work}
\label{sec:related}

\paragraph{{\mixup}.}

See Section~\ref{sec:warmup}. {\mixup} can improve both generalization and adversarial robustness~\cite{vlbnmclb18}. It has also been adapted to various learning tasks, including semi-supervised data augmentation~\cite{bcgpor19}, unsupervised image synthesis~\cite{bhvlghbp19}, and adversarial defense at the inference stage~\cite{pxz19}. Recently, \cite{fwxmw19} combined {\mixup} with model aggregation to defend against inversion attack, and \cite{facemix19} proposed a novel method using {\mixup} for on-cloud privacy-preserving inference.

\paragraph{Differential privacy.} 

Differential privacy for deep learning involves controlling privacy leakage by adding noise to the learning pipeline. If the noise is drawn from certain distributions, say Gaussian or Laplace, it is possible to provide guarantees of privacy~\cite{dkmmn06,dr14}. Applying differential privacy techniques to distributed deep learning is non-trivial. Shokri and Shmatikov \cite{ss15} proposed a distributed learning scheme by directly adding noise to the shared gradients. However, the amount of privacy guaranteed drops with the number of training epochs and the size of shared parameters. DPSGD~\cite{acg+16} was proposed to dynamically keep track of privacy spending based on the composition theorem~\cite{d09}. However, it still leads to an accuracy drop of about $20\%$ on CIFAR-10 dataset. Also, to control privacy leakage, DPSGD has to start with a  model pre-trained using nonprivate labeled data, and then carefully fine-tunes a few layers using private data.

\paragraph{Privacy using cryptographic protocols.}  In distributed learning setting with multiple data participants, it is possible for the participants to jointly train a model over their private inputs by employing techniques like homomorphic encryption~\cite{g09,graepel2012ml,li2017multi} or secure multi-party computation (MPC)~\cite{yao82,b11,mpzy17,dglms19}. Recent work proposed to use cryptographic methods to secure federated learning by designing a secure gradients aggregation protocol~\cite{bonawitz2016practical} or encrypting gradients~\cite{ahwm17}.  These approaches  slow down the computation by orders of magnitude, and may also require special security environment setups.

\paragraph{Instance Hiding.} See Appendix~\ref{sec:hide_app}.

%% file: concl.tex
\section{Discussions of Potential Attacks}\label{sec:potential_attacks}

We have received proposals of attacks since the release of early versions of this manuscript. We would like to thank these comments, which helped enhance the security of {\instahide}. We hereby summarize some of them and explain why {\instahide} in its current design is not vulnerable to them.

\paragraph{Tramèr's attack.} Florian Tramèr~\cite{f20} suggested an attack pipeline which 1) firstly uses GANs to undo the random one-time sign flips of encrypted images, and 2) then runs standard image-similarity search between the recovered image and images in the public dataset. 

We have evaluated this attack scenario in Section~\ref{sec:exp_privacy} (see `Demask using GAN' for the first step, and `Uncover  public  images  by  similarity  search' for the second step). As shown, the GAN demasking step corrects about 1/4 of the flipped signs (see Figure~\ref{fig:GAN_demask}), however this doesn't appear enough to allow further attacks (see Figure~\ref{fig:public_attack}). 

Several designs of {\instahide} could provide better practice of privacy under this attack:
1) setting a conservative upper threshold for the coefficients in mixing (e.g. 0.65 in our experiments) makes it harder to train the demasking GAN. 2) To alleviate the threat of image-similarity search for public image retrieval, we use randomly cropped patches from public images. This could enlarge the search space of pubic images by some large constant factor. We also suggest using a very large public dataset (e.g. ImageNet) if possible.

\paragraph{Braverman's attack.}
Mark Braverman~\cite{b20} proposed another possible attack to retrieve public images got mixed in an {\instahide} image. Given $\tilde{x} \in \R^d$, the {\instahide} image and the public dataset ${\cal D} \subset \R^d$, the attacker calculates for each $s \in \cal D$ the value $v_{s} = \langle \tilde{x}^2, s^2\rangle-\frac{1}{d}\|\tilde{x}\|_2^2 \|s\|_2^2$, and ranks them in descending order.  Braverman suggested that in theory, if images are drawn from Gaussian distributions, the public images got mixed will have higher ranks (i.e., there is a gap between $v_{s}$'s for $s$'s got mixed in  $\tilde{x}$ and $v_{s}$'s for $s$'s not got mixed in $\tilde{x}$). With that, the attacker may be able to shrink the candidate search space for public images by a large factor. A brief accompanying paper with theory and experiments for this attack is forthcoming.

Braverman's attack can be seen as similar in spirit to the attack on vanilla Mixup in Section~\ref{sec:mixup_attack_public}, with the difference that it needs images to behave like Gaussian vectors for higher moments, and not just inner products. In the  empirical study, we find 1) the gaussianity assumption is not so good for real images 2) the shrinkage of candidate search space by Braverman's attack is small for $k=4$ and almost disappears for $k =6$. 

To alleviate the risk of Bravermen's attack, we reemphasize our suggestions in Section~\ref{sec:privacy_suggestion}: 1) use a large public dataset and use random cropping as the preprocessing step. This will increase the original search space of public images; 2) use a larger $k$ (e.g. 6) if possible (3) maybe use a high quality GAN instead of the public dataset for mixing.  

\paragraph{Carlini et al.'s attack.} Carlini et al. \cite{carlini_attack} gave an attack to recover private images in the most vulnerable application of {\instahide} - when the {\instahide} encryptions are revealed to the attacker. The first step is to train a neural network on a public dataset for similarity annotation to infer whether a pair of {\instahide} encryptions contain the same private image. With the inferred similarities of all pairs of encryptions, the attacker then runs a combinatorial algorithm (cubic time in size of private dataset) to cluster all encryptions based on their  original private images, and finally uses a regression algorithm to recover the private images. 

Carlini et al.'s attack is able to give high-fidelity recoveries on our challenge dataset of 100 private images, but several limitations may prevent the current attack from working in a more realistic setting:
\begin{itemize}
    \item The current attack runs in time \textbf{cubic} in the dataset size, and it can’t directly attack an individual encryption. Running time was not an issue for attacking our challenge set, which consisted of 5,000 encrypted images derived from 100 images. But feasibility on larger datasets becomes challenging. %\Yang{Shall we provide a time estimation for larger datasets here?} %  The initial quadratic step takes them 0.1 GPU hours on the challenge set, but this should balloon to 2000+ GPU hours for even CIFAR10 — a modest dataset of 50k images. The math with cubic runtime and/or larger datasets is worse.
    \item The challenge dataset corresponded to an ambitious form of security, where the encrypted images themselves are released to the world. The more typical application is a Federated Learning scenario where the attacker observes gradients computed using the inputs. When {\instahide} is adopted in Federated Learning, the attacker only observes gradients computed on encrypted images. The attacks in this paper do not currently apply to that scenario. A possible attack to recover the original images may require recovering encrypted images from the gradients \cite{zlh19, geiping2020inverting} as the first step. 
\end{itemize}

\section{Conclusion}
\label{sec:conclude}
{\instahide} is a practical instance-hiding method for image data for private distributed deep learning.  

{\instahide} uses the {\mixup} method with a one-time secret key consisting of a pixel-wise random sign-flipping mask and samples from the same training dataset (Inside-dataset 
{\instahide})
or a large public dataset (Cross-dataset {\instahide}).  The proposed method can be easily plugged into any existing distributed learning pipeline.  It is very efficient and incurs minor reduction in accuracy. Maybe modifications of this idea can further alleviate the loss in accuracy that we observe, especially as $k$ increases.

We hope our analysis of {\instahide}'s security on worst-case vectors will motivate further theoretical study, including for average-case settings and for adversarial robustness. In Appendix~\ref{sec:app_phase_retrieval}, we suggest that although {\instahide} can be formulated as a phase retrieval problem, classical techniques have failed as attacks. 

We have tried statistical and computational attacks against {\instahide} without success. To encourage other researchers to try new attacks, we release a challenge dataset of encrypted images.  

%% file: hide_app.tex
The appendix is organized as follows: Appendix~\ref{sec:hide_app} reviews Instance hiding. Appendix~\ref{sec:attack_app} provides more details for two attacks on {\mixup} schemes in  Section~\ref{sec:warmup}. 
Appendix~\ref{sec:hard_app} discusses $k$-SUM, a well-known fine-grained complexity problem that is related to the worst-case security argument of {\instahide}. Appendix~\ref{sec:app_phase_retrieval} shows the connection between {\instahide} and phase retrieval. Finally, Appendix~\ref{sec:exp_app} provides experimental details.

\section{Instance Hiding}\label{sec:hide_app}

In the classical setting of instance hiding~\cite{afk87} in cryptography,  a computationally-limited Alice is trying to get more powerful computing services Bob$_1$ and Bob$_2$ to help her compute a function $f$ on input $x$, without revealing $x$. The simplest case is that $f$ is a linear function  over a finite field (e.g., integers modulo a prime number). Then  Alice can pick a random number $r$ and ``hide" the input $x$ by asking Bob$_1$ for $f(x + r)$ and Bob$_2$ for $f(r)$, and then infer $f(x)$ from the two answers. When all arithmetic is done modulo a prime, it can be shown that  neither Bob$_1$ nor Bob$_2$ individually learns anything (information-theoretically speaking) about $x$. This scheme can also be applied to compute polynomials instead of linear functions. 

{\instahide} is inspired by the special case where there is a single computational agent Bob$_1$. Alice has to use random values $r$ such that she  knows $f(r)$, and simply ask Bob$_1$ to supply $f(x+r)$. 
Note that such a random value $r$ would be {\em use-once} (also called {\em nonce} in cryptography); it would not be reused when trying to evaluate a different input. 

%% file: attack_app.tex
\section{Attacks on {\mixup}}\label{sec:attack_app}

Here we provide more details for the attacks discussed in Section~\ref{sec:warmup}.
\paragraph{Notations.} 

We use $\langle u,v \rangle$ to denote the inner product between vector $u$ and $v$. We use $\mathbf{0}$ to denote the zero vector. For a vector $x$, we use $\| x \|_2$ to denote its $\ell_2$ norm. For a positive integer $n$, we use $[n]$ to denote set $\{1,2,\cdots,n\}$. We use $\Pr[]$ to denote the probability, use $\E[]$ to denote the expectation. We use ${\cal N}(\mu,\sigma^2)$ to denote Gaussian distribution. For any two vectors $x,y$, we use $\langle x, y \rangle$ to denote the inner product.

\subsection{Don't mix up the same image multiple times}
Let us continue with the vision task. This attack argues that given a pair of {\mixup} images $\tilde{x}_1$ and $\tilde{x}_2$, by simply checking $\langle \tilde{x}_1, \tilde{x}_2\rangle$, the attacker can determine with high probability whether $\tilde{x}_1$ and $\tilde{x}_2$ are derived from the same image. We show this by a simple case with $k = 2$ in Theorem~\ref{thm:multi_mix}, where $k$ is the number of images used to generate a {\mixup} sample.

\begin{theorem}\label{thm:multi_mix}
Let $\mathcal{X} \subset \R^d$ with $|{\cal X}| = n$ and $\forall x \in \mathcal{X}, $ we sample $x \sim {\cal N}(0, \sigma^2 I)$ where $\sigma^2 = 1/d$. Let $\mathcal{X}_1$, $\mathcal{X}_2$ and ${\mathcal{X}}_3$ denote three disjoint sets such that ${\mathcal{X}}_1 \cup {\mathcal{X}}_2 \cup {\mathcal{X}}_3 = \mathcal{X}$, with probability $1-\delta$, we have: \\
{\bf Part 1.}  For $x_1\in {\cal X}_1, x_2, x_2' \in {\cal X}_2, x_3 \in {\cal X}_3$, $\langle {x}_3+x_1,{x}_2+x_2' \rangle \leq c_1 \cdot 4\sigma^2 \sqrt{d} \log^2(d n /\delta)$. \\
{\bf Part 2.}  For $x_1 \in {\cal X}_1, x_2 \in {\cal X}_2, x_3 \in {\cal X}_3$, $\langle {x}_3 + {x}_1, {x}_3 + {x}_2 \rangle \geq (c_2 \cdot \sqrt{\log(n/\delta)}+d) \sigma^2 -  c_1 \cdot 3\sigma^2 \sqrt{d} \log^2(d n /\delta)$.
where $c_1,c_2> 0$ are two universal constants.
\end{theorem}

\begin{proof}

{\bf Part 1.} 
First, we can expand $\langle {x}_3+x_1,{x}_2+x_2' \rangle$,
\begin{align*}
    \langle {x}_3+x_1,{x}_2+x_2' \rangle &= \langle {x}_3, {x}_2 \rangle + \langle {x}_1, {x}_2 \rangle + \langle {x}_3, {x}_2' \rangle + \langle {x}_1, {x}_2' \rangle 
\end{align*}

For each fixed $u \in \{ x_3, x_1 \}$ and each fixed $v \in \{ x_2, x_2' \}$, using Lemma~\ref{lem:concentration_of_inner_product_two_guassian}, we have
\begin{align*}
    \Pr \Big[ |\langle u , v \rangle | \geq c_1 \cdot \sigma^2 \sqrt{d} \log^2(d n /\delta) \Big] \leq \delta/n^2.
\end{align*}
where $c_1>1$ is some sufficiently large constant.

Since $x_1, x_2, x_2, x_3$ are independent random Gaussian vectors, taking a union bound over all pairs of $u$ and $v$, we have Lemma~\ref{lem:concentration_of_inner_product_two_guassian}, we have
\begin{align}
    \Pr[|\langle {x}_3+x_1,{x}_2+x_2' \rangle | \geq c_1 \cdot 4\sigma^2 \sqrt{d} \log^2(d n /\delta) \Big] \leq 4\delta/n^2.
\end{align}

{\bf Part 2.} 

We can lower bound $|\langle {x}_3+x_1,{x}_3+x_2 \rangle|$ in the following sense,
\begin{align*}
    |\langle {x}_3+x_1,{x}_3+x_2 \rangle| 
    = & ~ |\langle x_3, x_3 \rangle + \langle x_3, x_2 \rangle + \langle x_1, x_3 \rangle + \langle x_1, x_2 \rangle| \\
    \geq & ~ |\langle x_3, x_3 \rangle| - |\langle x_3, x_2 \rangle + \langle x_1, x_3 \rangle + \langle x_1, x_2 \rangle|
\end{align*}

For a fixed $x_3$, we can lower bound $\|x_3\|_2^2$ with Lemma~\ref{lem:chi_square_tail},
\begin{align*}
    \Pr \Big[ \| x_3 \|_2^2 > ( c_2 \cdot \sqrt{\log(n/\delta)}+d) \sigma^2 \Big] \geq 1 -\delta  / n^2 .
\end{align*}

Since $x_1, x_2, x_3$ are independent random Gaussian vectors, using Lemma~\ref{lem:concentration_of_inner_product_two_guassian}, we have
\begin{align}
    \Pr[|\langle x_3, x_2 \rangle + \langle x_1, x_3 \rangle + \langle x_1, x_2 \rangle| \geq c_1 \cdot 3\sigma^2 \sqrt{d} \log^2(d n /\delta) \Big] \leq 3\delta/n^2.
\end{align}
where $c_1 > 1$ is some sufficiently large constant.

Thus, for a fixed $x_3$, we have
\begin{align*}
     |\langle {x}_3+x_1,{x}_3+x_2 \rangle| \geq  ( c_2 \cdot \sqrt{\log(n/\delta)}+d) \sigma^2 -  c_1 \cdot 3\sigma^2 \sqrt{d} \log^2(d n /\delta)
\end{align*}
holds with probability $1-\delta/n^2-3\delta/n^2 = 1 - 4\delta /n^2$.
\end{proof}

\begin{corollary}\label{corollary:no_multi_x}
 Let $\mathcal{X} \subset \R^d$ with $|{\cal X}| = n$ and $\forall x \in \mathcal{X}, $ we sample $x \sim {\cal N}(0, \sigma^2 I)$ where $\sigma^2 = 1/d$. Let $\mathcal{X}_1$, $\mathcal{X}_2$ and ${\mathcal{X}}_3$ denote three disjoint sets such that ${\mathcal{X}}_1 \cup {\mathcal{X}}_2 \cup {\mathcal{X}}_3 = \mathcal{X}$.\\
 For any $\beta > 1$, if $ (2\beta)^{-1} \cdot \sqrt{d} \log^2 (nd/\delta) \geq 4$, then with probability $1-\delta$ , we have :\\ for $x_1 \in \mathcal{X}_1, x_2, x_2' \in \mathcal{X}_2, x_3 \in \mathcal{X}_3$,
 \begin{align*}
    |\langle x_3+x_1, x_3+x_2 \rangle| \geq \beta \cdot |\langle x_3+x_1, x_2+x_2' \rangle |.
 \end{align*}
 \end{corollary}
\begin{proof}

With Theorem~\ref{thm:multi_mix}, we have with probability $1-\delta$, we have : 
\begin{align*}
    \frac{|\langle x_3+x_1, x_3+x_2 \rangle|}{|\langle x_3+x_1, x_2+x_2' \rangle|}  
    \geq & ~ \frac{(c_2 \cdot \sqrt{\log(n/\delta)}+d) \sigma^2 -  c_1 \cdot 3\sigma^2 \sqrt{d} \log^2(d n /\delta)}{c_1 \cdot 4 \sigma^2 \sqrt{d} \log^2 (nd/\delta)}\\
    = & ~ \frac{c_2 \cdot \sqrt{\log(n/\delta)}+d - c_1 \cdot 3  \sqrt{d} \log^2 (nd/\delta)}{c_1 \cdot 4  \sqrt{d} \log^2 (nd/\delta)}\\
    = & ~ \frac{1 + c_2 / (c_1\cdot d) \cdot \sqrt{\log(n/\delta)} - 3 / \sqrt{d} \cdot \log^2 (nd/\delta)}{4 /  \sqrt{d} \cdot \log^2 (nd/\delta)} \\
    \geq & ~ \frac{1  - 3 / \sqrt{d} \cdot \log^2 (nd/\delta)}{4 /  \sqrt{d} \cdot \log^2 (nd/\delta)} 
    \geq ~ \frac{1 - 1/(2\beta) }{ 1/(2\beta) } 
    =  ~ 2 \beta -1 
    \geq  ~ \beta,
\end{align*}
where the forth step follows from choice $c_1$ and $c_2$, the fifth step follows from assumption in Lemma statement, and the last step follows from $\beta > 1$.

Thus, we complete the proof.
\end{proof}

\subsection{Attacks that run in \texorpdfstring{$|\mathcal{X}|$}{} time}
This attack says that, if a cross-dataset {\mixup} sample  $\tilde{x}$ is generated by mixing 1 sample from a privacy-sensitive original dataset and $k-1$ samples from a public dataset (say ImageNet \cite{imagenet09}), then the attacker can crack the $k-1$ samples from the public dataset by simply checking the inner product between $\tilde{x}$ and all images in the public dataset. We show this formally in Theorem~\ref{thm:attack_in_X}.

\begin{theorem}\label{thm:attack_in_X}
 Let $\mathcal{X} \subset \R^d$ with $|{\cal X}| = n$ and $\forall x \in \mathcal{X}, $ we sample $x \sim {\cal N}(0, \sigma^2 I)$. Let $\tilde{x} = \sum_{i=1}^k x_i$, where $x_i \in \mathcal{X}$, with probability $1-\delta$, we have:\\
{\bf Part 1.} For all $t' \in [n] \backslash [k]$, $\langle \tilde{x}, x_{t'} \rangle \leq c_1 \cdot k \sigma^2 \sqrt{d} \log^2 (nd/\delta)$.\\
{\bf Part 2.} For all $t \in [k]$, $\langle \tilde{x}, x_t \rangle \geq (c_2 \cdot \sqrt{\log(n/\delta)}+d)\sigma^2 - c_1 \cdot (k-1) \sigma^2 \sqrt{d} \log^2 (nd/\delta)$. \\
where $c_1,c_2> 0$ are two universal constants.
\end{theorem}

We remark that the proof of Theorem~\ref{thm:attack_in_X} is similar to the proof of Theorem~\ref{thm:multi_mix}.
\begin{proof}

{\bf Part 1.} 
We can rewrite $\langle \tilde{x},x_{t'} \rangle$ as follows:
\begin{align*}
    \langle \tilde{x}, x_{t'} \rangle &= \sum_{i=1}^k \langle {x}_i, {x}_{t'} \rangle,
\end{align*}
which implies
\begin{align*}
    |\langle \tilde{x}, x_{t'} \rangle| \leq \sum_{i=1}^k | \langle {x}_i, {x}_{t'} \rangle |.
\end{align*}

For each fixed $i \in [k]$ and each fixed $t'\notin [k]$, using Lemma~\ref{lem:concentration_of_inner_product_two_guassian}, we have
\begin{align*}
    \Pr \Big[ |\langle x_i , x_{t'} \rangle | \geq c_1 \cdot \sigma^2 \sqrt{d} \log^2(d n /\delta) \Big] \leq \delta/n^2.
\end{align*}
where $c_1>1$ is some sufficiently large constant.

Since $x_1, x_2, \cdots, x_k, x_{t'}$ are independent random Gaussian vectors, taking a union over all $i \in [k]$, we have
\begin{align*}
    \Pr \Big[ |\langle \tilde{x} , x_{t'} \rangle | \geq c_1 \cdot k\sigma^2 \sqrt{d} \log^2(d n /\delta) \Big] \leq \delta k/n^2 .
\end{align*}

Taking a union bound over all $t' \in [n] \backslash [k]$, we have
\begin{align*}
    \Pr \Big[ \forall t' \in [n] \backslash [k], ~~~ |\langle \tilde{x} , x_{t'} \rangle | \geq c_1 \cdot k\sigma^2 \sqrt{d} \log^2(d n /\delta) \Big] \leq \delta (n-k)k/n^2 .
\end{align*}

{\bf Part 2.} 

We can lower bound $|\langle \tilde{x},x_t \rangle |$ as follows:
\begin{align*}
    |\langle \tilde{x},x_t \rangle |
    = & ~ 
    \Big| \langle {x}_t, {x}_t \rangle + \sum_{i \in [k] \backslash \{t\}} \langle {x}_i, {x}_t \rangle \Big| 
    \geq  |\langle {x}_t, {x}_t \rangle | - \Big| \sum_{i \in [k] \backslash \{t\}} \langle {x}_i, {x}_t \rangle \Big| \\
    \geq & ~ |\langle {x}_t, {x}_t \rangle | -  \sum_{i \in [k] \backslash \{t\}} | \langle {x}_i, {x}_t \rangle |.
\end{align*}
First, we can bound $\|x_t\|_2^2$ with Lemma~\ref{lem:chi_square_tail},
\begin{align*}
    \Pr \Big[ \| x_t \|_2^2 > ( c_2 \cdot \sqrt{\log(n/\delta)}+d) \sigma^2 \Big] \geq 1 -\delta  / n^2 .
\end{align*}
For each $i \in [k] \backslash \{t\}$, using Lemma~\ref{lem:concentration_of_inner_product_two_guassian}, we have
\begin{align*}
    \Pr \Big[ |\langle x_i , x_{t} \rangle | \geq c_1 \cdot \sigma^2 \sqrt{d} \log^2(d n /\delta) \Big] \leq \delta/n^2.
\end{align*}
where $c_1 > 1$ is some sufficiently large constant.

Taking a union bound over all $i \in  [k] \backslash \{t\}$, we have
\begin{align*}
    \Pr \Big[ \sum_{i \in [k] \backslash \{t\} } |\langle x_i , x_{t} \rangle | \geq c_1 \cdot (k-1)\sigma^2 \sqrt{d} \log^2(d n /\delta) \Big] \leq \delta (k-1)/n^2.
\end{align*}

Thus, for a fixed $t\in [k]$, we have
\begin{align*}
     |\langle \tilde{x},x_t \rangle | \geq  ( c_2 \cdot \sqrt{\log(n/\delta)}+d) \sigma^2 -  c_1 \cdot (k-1)\sigma^2 \sqrt{d} \log^2(d n /\delta)
\end{align*}
holds with probability $1-\delta/n^2-\delta(k-1)/n^2 = 1 - \delta k /n^2$.

Taking a union bound over all $t \in [k]$, we complete the proof.
\end{proof}

\begin{corollary}\label{corollary:attacks_in_X}
 Let $\mathcal{X} \subset \R^d$ with $|{\cal X}| = n$ and $\forall x \in \mathcal{X}, $ we sample $x \sim {\cal N}(0, \sigma^2 I)$. Let $\tilde{x} = \sum_{i=1}^k x_i$, where $x_i \in \mathcal{X}$.\\
 For any $\beta > 1$, if $k \leq (2\beta)^{-1} \cdot \sqrt{d} \log^2 (nd/\delta)$, then with probability $1-\delta$ , we have :\\ for all $t' \in [n] \backslash [k]$ and all $t \in [k]$
 \begin{align*}
    |\langle \tilde{x}, x_{t} \rangle| \geq \beta \cdot |\langle \tilde{x}, x_{t'} \rangle |.
 \end{align*}
\end{corollary}
\begin{proof}
With Theorem~\ref{thm:attack_in_X}, we have with probability $1-\delta$, we have : for all $t\in [k]$ and $t'\notin [k]$,
\begin{align*}
    \frac{|\langle \tilde{x}, x_{t} \rangle|}{|\langle \tilde{x}, x_{t'} \rangle|}  
    \geq & ~ \frac{(c_2 \cdot \sqrt{\log(n/\delta)}+d)\sigma^2 - c_1 \cdot (k-1) \sigma^2 \sqrt{d} \log^2 (nd/\delta)}{c_1 \cdot k \sigma^2 \sqrt{d} \log^2 (nd/\delta)}\\
    = & ~ \frac{c_2 \cdot \sqrt{\log(n/\delta)}+d - c_1 \cdot (k-1)  \sqrt{d} \log^2 (nd/\delta)}{c_1 \cdot k  \sqrt{d} \log^2 (nd/\delta)}\\
    = & ~ \frac{1 + c_2 / (c_1\cdot d) \cdot \sqrt{\log(n/\delta)} - (k-1) / \sqrt{d} \cdot \log^2 (nd/\delta)}{k /  \sqrt{d} \cdot \log^2 (nd/\delta)} \\
    \geq & ~ \frac{1  - k / \sqrt{d} \cdot \log^2 (nd/\delta)}{k /  \sqrt{d} \cdot \log^2 (nd/\delta)} 
    \geq ~ \frac{1 - 1/(2\beta) }{ 1/(2\beta) } 
    =  ~ 2 \beta -1 
    \geq  ~ \beta,
\end{align*}
where the forth step follows from the choice of $c_1$ and $c_2$, the fifth step follows from the assumption in the Lemma statement, and the last step follows from $\beta > 1$.

Thus, we complete the proof.
\end{proof}

\subsection{Chi-square concentration and Bernstein inequality}
We state two well-known probability tools in this section.\footnote{The major of idea of the provable results in this section is $p$-th moment concentration inequality (see Lemma 3.1 in \cite{cls20} as an example). The inner product ``attack'' for {\mixup} is based on $p=2$. The inner product ``attack'' for {\instahide} is based on $p=4$ (suggested by \cite{b20}).} One is the concentration inequality for Chi-square and the other is Bernstein inequality.

First, we state a concentration inequality for Chi-square:
\begin{lemma}[Lemma 1 on page 1325 of Laurent and Massart \cite{lm00}]
\label{lem:chi_square_tail}
Let $X \sim {\cal X}_k^2$ be a chi-squared distributed random variable with $k$ degrees of freedom. Each one has zero mean and $\sigma^2$ variance. Then
\begin{align*}
\Pr[ X - k \sigma^2 \geq ( 2 \sqrt{kt} + 2t ) \sigma^2 ] \leq \exp (-t), \mathrm{~~~and~~~}
\Pr[ k \sigma^2 - X \geq 2 \sqrt{k t} \sigma^2 ] \leq \exp(-t).
\end{align*}
\end{lemma}

We state the Bernstein inequality as follows:
\begin{lemma}[Bernstein inequality \cite{b24}]\label{lem:bernstein}
Let $X_1, \cdots, X_n$ be independent zero-mean random variables. Suppose that $|X_i| \leq M$ almost surely, for all $i \in [n]$. Then, for all $t > 0$,
\begin{align*}
\Pr \left[ \sum_{i=1}^n X_i > t \right] \leq \exp \left( - \frac{ t^2/2 }{ \sum_{j=1}^n \E[X_j^2]  + M t /3 } \right).
\end{align*}
\end{lemma}

\subsection{Inner product between a random Gaussian vector a fixed vector}

The goal of this section is to prove Lemma~\ref{lem:concentration_of_inner_product_guassian}. It provides a high probability bound for the absolute value of inner product between one random Gaussian vector with a fixed vector.
\begin{lemma}[Inner product between a random Gaussian vector and a fixed vector]\label{lem:concentration_of_inner_product_guassian}
    Let $u_1, \cdots, u_d$ denote i.i.d. random Gaussian variables where $u_i \sim {\cal N}(0,\sigma_1^2)$. 
        
Then, for any fixed vector $e \in \R^d$, for any failure probability $\delta \in (0,1/10)$, we have
\begin{align*}
    \Pr_{u} \Big[ | \langle u, e \rangle | \geq  2 \sigma_1 \|e\|_2 \sqrt{ \log(d/\delta)} + \sigma_1\|e\|_{\infty}\log^{1.5}(d/\delta) \Big] \leq \delta.
\end{align*}
\end{lemma}

\begin{proof}
        
First, we can compute $\E[ u_i ]$ and $\E[u_i^2]$
\begin{align*}
    \E[ u_i ] = 0, \text{~~~and~~~} \E[ u_i^2 ] = \sigma_1^2.
\end{align*}

Next, we can upper bound $|u_i|$ and $|u_i e_i|$.
\begin{align*}
    \Pr_{u}[|u_i-\E[u_i]| \geq t_1] \leq \exp ( - t_1^2 / ( 2\sigma_1^2 )  ) .
\end{align*}
Take $t_1 = \sqrt{2\log(d/\delta)} \sigma_1$, then for each fixed $i\in [d]$, we have, $|u_i| \leq \sqrt{2\log(d/\delta)}\sigma_1$ holds with probability $1 - \delta/d$.

Taking a union bound over $d$ coordinates, with probability $1 - \delta$, we have : for all $i \in [d]$, $|u_i| \leq \sqrt{ 2 \log( d / \delta ) } \sigma_1$.

Let $E_1$ denote the event that, $\max_{i\in[d]}|u_ie_i|$ is upper bounded by $\sqrt{2\log(d/\delta)}\sigma_1\|e\|_{\infty}$. $\Pr[E_1] \geq 1-\delta$.
    
Using Bernstein inequality (Lemma~\ref{lem:bernstein}), we have
\begin{align*}
    \Pr_{u}[ | \langle u , e \rangle | \geq t ] 
    \leq & ~ \exp \Big( -\frac{t^2/2}{ \| e\|_2^2 \E[u_i^2] +  \max_{i\in[d]} |u_i e_i| \cdot t/3 } \Big) \\
    \leq & ~ \exp \Big( -\frac{t^2/2}{ \| e\|_2^2 \sigma_1^2 + \sqrt{2\log(d/\delta)}\sigma_1\|e\|_{\infty}\cdot t/3 } \Big) 
    \leq ~ \delta,
\end{align*}
where the second step follows from $\Pr[E_1] \geq 1-\delta$ and $\E[u_i^2] = \sigma_1^2$, and the last step follows from choice of $t$:
\begin{align*}
    t = 2\sigma_1 \|e\|_2 \sqrt{ \log(d / \delta) } + \sigma_1\|e\|_{\infty}\log^{1.5}(d/\delta)
\end{align*}
       
Taking a union bound with event $E_1$, we have $\Pr[ |\langle u, e \rangle| \geq t ] \leq 2 \delta$. 
Finally, rescaling $\delta$ finishes the proof.
\end{proof}
    
\subsection{Inner product between two random Gaussian vectors}

The goal of this section is to prove Lemma~\ref{lem:concentration_of_inner_product_two_guassian}. It provides a high probability bound for the absolute value of inner product between two random (independent) Gaussian vectors. 
\begin{lemma}[Inner product between two random Gaussian vectors]\label{lem:concentration_of_inner_product_two_guassian}
Let $u_1, \cdots, u_d$ denote i.i.d. random Gaussian variables where $u_i \sim {\cal N}(0,\sigma_1^2)$ and $e_1, \cdots, e_d$ denote i.i.d. random Gaussian variables where $e_i \sim {\cal N}(0,\sigma_2^2)$.
    
Then, for any failure probability $\delta \in (0,1/10)$, we have
\begin{align*}
    \Pr_{u,e} \Big[ | \langle u, e \rangle | \geq 10^4 \sigma_1 \sigma_2 \sqrt{d} \log^2(d/\delta) \Big] \leq \delta.
\end{align*}
\end{lemma}

\begin{proof}
First, using Lemma~\ref{lem:chi_square_tail}, we compute the upper bound for $\| e \|_2^2$
\begin{align*}
    \Pr_{e}[\|e\|_2^2 - d\sigma_2^2 \geq (2\sqrt{dt}+2t)\sigma_2^2] \leq \exp(-t).
\end{align*}
    
Take $t = \log(1/\delta)$, then with probability $1-\delta$, 
\begin{align*}
    \|e\|_2^2 \leq (d+3\sqrt{d\log(1/\delta)}+2\log(1/\delta))\sigma_2^2 \leq 4d \log(1/\delta) \sigma_2^2.
\end{align*}
Thus
\begin{align*}
    \Pr_{e} [ \|e\|_2 \leq 4 \sqrt{d \log(1/\delta)} \sigma_2 ] \geq 1- \delta.
\end{align*}
     
Second, we compute the upper bound for $\|e\|_{\infty}$ (the proof is similar to Lemma~\ref{lem:concentration_of_inner_product_guassian})
\begin{align*}
    \Pr_{e}[|\|e\|_{\infty} \leq \sqrt{\log(d/\delta)} \sigma_2] &\geq 1- \delta.
\end{align*}
    
We define $t$ and $t'$ as follows
\begin{align*}
    t= ~ 4 \cdot ( \sigma_1 \|e\|_2 \sqrt{ \log(d/\delta)} + \sigma_1\|e\|_{\infty}\log^{1.5}(d/\delta) ), \text{~~~and~~~}
    t' = ~ 8 \cdot ( \sigma_1\sigma_2\sqrt{d}\log(d/\delta) + \sigma_1\sigma_2 \log^2(d/\delta) )
\end{align*}
From the above calculations, we can show
\begin{align*}
    \Pr_{e}[t' \geq t] \geq 1 - 2\delta.
\end{align*}

By Lemma \ref{lem:concentration_of_inner_product_guassian}, for fixed $e$, we have
\begin{align*}
    \Pr_{u}[ | \langle u , e \rangle | \geq t ] \leq  \delta
\end{align*}
          
Overall, we have
\begin{align*}
    \Pr_{e,u}[ | \langle u , e \rangle | \geq t' ] \leq  3\delta
\end{align*}
Therefore, rescaling $\delta$ completes the proof.
\end{proof}

%% file: hard_app.tex
\section{Computational hardness results of \texorpdfstring{$d$}{}-dimensional \texorpdfstring{$k$}{}-SUM}\label{sec:hard_app}

The basic components in InstaHide schemes are inspired by computationally hard problems derived from the classic SUBSET-SUM problem: given a set of integers, decide if there is a non-empty subset of integers whose integers sum to $0$. It is a version of knapsack, one of the Karp's 21 NP-complete problems \cite{k72}. The $k$-SUM \cite{e95} is the parametrized version of the SUBSET-SUM. Given a set of integers, one wants to ask if there is a subset of $k$ integers sum to $0$. The $k$-SUM problem can be solved in $O(n^{\lceil k/2 \rceil})$ time. For any integer $k \geq 3$ and constant $\epsilon >0$, whether $k$-SUM can be solved in $O(n^{\lfloor k / 2 \rfloor - \epsilon})$ time has been a long-standing open problem. Patrascu \cite{p10}, Abboud and Lewi \cite{al13} conjectured that such algorithm doesn't exist. 

It is natural to extend definition from one-dimensional scalar/number case to the high-dimensional vector case. For $d$-dimensional $k$-sum over finite field, Bhattacharyya, Indyk, Woodruff and Xie \cite{biwx11} have shown that any algorithm that solves this problem has to take $\min \{ 2^{\Omega(d)} , n^{\Omega(k)} \}$ time unless Exponential Time Hypothesis is false. Here, Exponential Time Hypothesis is believed to be true, it states that there is no $2^{o(n)}$ time to solve 3SAT with $n$ variables. 

In this work, we observe that privacy of InstaHide can be interpreted as $d$-dimensional $k$-SUM problem and thus could be intractable and safe. For real field, $d$-dimensional is equivalent to $1$-dimensional due to \cite{alw14}. Therefore, Patrascu \cite{p10}, Abboud and Lewi \cite{al13}'s conjecture also suitable for $d$-dimensional $k$-SUM problem, and several hardness results in $d=1$ also can be applied to general $d>1$ directly.

\subsection{\texorpdfstring{$k$}{}-SUM}

We hereby provide a detailed explanation for the $d$-dimensional $k$-SUM problem. Let us start with the special case of $d=1$ and all values are integers.

\begin{definition}[SUBSET-SUM]
Given a set of integers, if there is a subset of integers whose integers sum to $0$.
\end{definition}
SUBSET-SUM is a well-known NP-complete problem.

The $k$-SUM is the parameterized version of the SUBSET-SUM,
\begin{definition}[$k$-SUM]
Given a set of integers, if there is a subset of $k$ integers whose integers sum to $0$.
\end{definition}

The $k$-SUM problem can be solved in $O(n^{ \lceil k / 2 \rceil } )$ time. For $k=3$, Baran, Demaine and Patrascu \cite{bdp08} introduced algorithm that takes $O(n^2 / \log^2 n)$ time. It has been a longstanding open problem to solve $k$-SUM for some $k$ in time $O( n^{ \lceil k / 2 \rceil - \epsilon } )$. Therefore, complexity communities made the following conjecture,

\begin{conjecture}[The $k$-SUM Conjecture, \cite{al13}]
There does not exist a $k \geq 2$, an $\epsilon > 0$, and a randomized algorithm that succeeds (with high probability) in solving $k$-{\rm SUM} in time $O( n^{ \lceil k / 2 \rceil - \epsilon } )$.
\end{conjecture}

Although the $n^{\lceil(k/2) \rceil}$ hardness for $k$-SUM is not based on anything else at the moment, an $n^{\Omega(k)}$ lower bound under ETH is already known due to Abboud and Lewi \cite{al13}. Recently, Abboud \cite{a19} also shows a weaker $ n^{ \Omega ( k / \log k ) }$ lower bound under the Set Cover Conjecture, but it has the advantage that it holds for any fixed $k>2$.

\subsection{\texorpdfstring{$k$}{} Vector Sum over Finite Field}

Now we move onto the $d$-dimensional $k$ vector sum problem.

\begin{definition}[$d$-dimensional $k$-VEC-SUM over finite field]\label{def:d_dim_k_sum_finite}
Given a set of $n$ vectors in $\mathbb{F}_q^d$, if there is a subset of $k$ vectors such that the summation of those $k$ vectors is an all $0$ vector.
\end{definition}

A more general definition that fits the InstaHide setting is called $k$-VEC-T-SUM (where $T$ denotes the ``target''),
\begin{definition}[$d$-dimensional $k$-VEC-T-SUM over finite field]\label{def:d_dim_k_t_sum_finite}
Given a set of $n$ vectors in $\mathbb{F}_q^d$, if there is a subset of $k$ vectors such that the summation of those $k$ vectors is equal to some vector $z$ in $\mathbb{F}_q^d$.
\end{definition}

$d$-dimensional $k$-VEC-T-SUM and $d$-dimensional $k$-VEC-SUM are considered to have the same hardness. Bhattacharyya, Indyk, Woodruff and Xie \cite{biwx11} proved hardness result for the problem defined in Definition~\ref{def:d_dim_k_sum_finite} and Definition~\ref{def:d_dim_k_t_sum_finite}. Before stating the hardness result, we need to define several basic concepts in complexity. We introduce the definition of 3SAT and Exponential Time Hypothesis(ETH). For the details and background of 3SAT problem, we refer the readers to \cite{ab09}.

\begin{definition}[3SAT problem]\label{def:3SAT}
Given $n$ variables and $m$ clauses conjunctive normal form {\rm CNF} formula with size of each clause at most $3$, the goal is to decide whether there exits an assignment for the $n$ boolean variables to make the {\rm CNF} formula be satisfied.
\end{definition}

We state the definition of Exponential Time Hypotheis, which can be thought of as a stronger assumption than P$\neq$NP.
\begin{hypothesis}[Exponential Time Hypothesis (ETH) \cite{ipz98}]\label{hyp:ETH}
There is a $\delta > 0$ such that {\rm 3SAT} problem defined in Definition~\ref{def:3SAT} cannot be solved in $O(2^{\delta n})$ running time.
\end{hypothesis}

Now, we are ready to state the hardness result.
\begin{theorem}[\cite{biwx11}]
  \label{thm:hardness_d_ksum}
Assuming Exponential Time Hypothesis ({\rm ETH}), any algorithm solves $k${\rm-VEC-SUM} or $k${\rm-T-VEC-SUM} requires $\min \{ 2^{\Omega(d)} , n^{\Omega(k)} \}$ time.
\end{theorem}

\subsection{\texorpdfstring{$k$}{} Vector Sum over Bounded Integers}

For the bounded integer case, we can also define the $k$-VEC-SUM problem,
\begin{definition}[$d$-dimensional $k$-VEC-SUM over bounded integers]\label{def:d_dim_k_sum_bounded}
For integers $k$, $n$, $M$, $d > 0$, the $k${\rm-VEC-SUM} problem is to determine, given vectors $x_1, \cdots, x_n, x \in [ 0 , k M ]^d$, if there is a size-$k$ subset $S \subseteq [n]$ such that $\sum_{i \in S} x_i = z$.
\end{definition}
Abboud, Lewi, and Williams proved that the 1-dimensional $k$-VEC-SUM and $d$-dimensional $k$-VEC-SUM are equivalent in bounded integer setting,
\begin{lemma}[Lemma 3.1 in \cite{alw14}]
Let $k,p,d,s,M \in \mathbb{N}$ satisfy $k < p$, $p^d \geq k M + 1$, and $s = (k+1)^{d-1}$. There is a collection of mappings $f_1, \cdot, f_s : [0,M] \times [0,kM] \rightarrow [-kp, kp]^d$, each computable in time $O( \poly \log M + k^d )$, such that for all numbers $x_1, \cdots, x_k \in [0,M]$ and targets $t \in [0,kM]$,
\begin{align*}
    \sum_{j=1}^k x_j = t \iff \exists i \in [s] \mathrm{~} \sum_{j=1}^k f_i(x_j, t) .
\end{align*}
\end{lemma}
Due to the above result, as long as we know a hardness result for classical $k$-SUM, then it automatically implies a hardness result for $d$-dimensional $k$-VEC-SUM.

%% file: phase_app.tex
\section{Phase Retrieval}\label{sec:app_phase_retrieval}

\subsection{Compressive Sensing}

Compressive sensing is a powerful mathematical framework the goal of which is to reconstruct an approximately $k$-sparse vector $x \in \R^n$ from linear measurements $y = \Phi x$, where $\Phi \in \R^{m \times n}$ is called ``sensing'' or ``sketching'' matrix. The mathematical framework was initiated by \cite{crt06,d06}.

We provide the definition of the $\ell_2/\ell_2$ compressive sensing problem.
\begin{definition}[Problem 1.1 in \cite{ns19}]
Given parameters $\epsilon, k, n$, and a vector $x \in \R^n$. The goal is to design some matrix $\Phi \in \R^{m \times n}$ and a recovery algorithm ${\cal A}$ such that we can output a vector $x'$ based on measurements $y = \Phi x$,
\begin{align*}
\| x' - x \|_2 \leq (1+\epsilon) \min_{k\mathrm{-sparse}~z \in \R^n} \| z - x \|_2.
\end{align*}
We primarily want to minimize $m$ (which is the number of measurements), the running time of ${\cal A}$(which is the decoding time) and column sparsity of $\Phi$.
\end{definition}
The $\ell_2/\ell_2$ is the most popular framework, and the state-of-the-art result is due to \cite{glps10,ns19}. 

\subsection{Phase Retrieval}

In compressive sensing, we can observe $y= \Phi x$. However, in phase retrieval, we can only observe $y = | \Phi x |$. Here, we provide the definition of $\ell_2/\ell_2$ phase retrieval problem.
\begin{definition}[\cite{ln18}]
Given parameters $\epsilon,k,n$, and a vector $x \in \R^n$. The goal is to design some matrix $\Phi \in \R^{m \times n}$ and a recovery algorithm ${\cal A}$ such that we can output a vector $x'$ based on measurements $y = | \Phi x |$,
\begin{align*}
    \| x' - x \|_2 \leq (1+\epsilon) \min_{k\mathrm{-sparse}~z \in \R^n} \| z - x \|_2.
\end{align*}
We primarily want to minimize $m$ (which is the number of measurements), the running time of ${\cal A}$ (which is the decoding time) and column sparsity of $\Phi$.
\end{definition}
The state-of-the-art result is due to \cite{ln18}.

The problem formulation of compressive sensing and phase retrieval have many variations, for more details we refer the readers to several surveys \cite{price13,h16,nakos19,song19}.

\subsection{Comments on Phase Retrieval}

We would like to point out that although {\instahide} can be formulated as a phase retrieval problem, classical techniques will fail as an attack. 

To show the phase retrieval formulation of {\instahide}, we first argue applying the random pixel-wise sign is equivalent to taking the absolute value, namely (A): $\title{x} = \sigma \circ f_{\text{mix}}(x)$ is equivalent to (B): $\title{x} = |f_{\text{mix}} ( x ) |$, where $f_{\text{mix}}(\cdot)$ denote the mixing function. This is because, to reduce from B to A, we can just take absolute value in each coordinate; and to reduce from A to B, we can just select a random mask (i.e., $\sigma$) and apply it.

Classical phase retrieval \cite{f78, f82, mrb07} aims to recover $x$ given $y = |\Phi x| $, where $\Phi \in \R^{m \times n}$, $x \in \R^{n}$ and $y \in \R^m$. In the {\instahide} setting, one can think of $\Phi$ as the concatenation of the private and public dataset, where $n = n_{\mathrm{private}}+n_{\mathrm{public}}$, and $x$ is a $k$-sparse vector which selects $k$ out of $n$ data points to {\mixup}. Therefore, $y = |\Phi x| $ is the `encrypted' sample using {\instahide}. It has been shown that solving $x$ can be formulated as a linear program \cite{mrb07}. It is well-known that linear program can be solved in polynomial in number of constraints/variables \cite{cls19,jswz20} and thus $m = O(k^2 \log(n/k^2))$ gives a polynomial-time optimization procedure. However, this result does not naturally serve as an attack on {\instahide}. First, $\Phi$ in phase retrieval needs to satisfy certain properties which may not be true for $\Phi$ in {\instahide}. More importantly, running LP becomes impossible for {\instahide} case, since part of $\Phi$ remains unknown.

%% file: exp_app.tex
\section{Experiments}\label{sec:exp_app}

\subsection{Network architecture and hyperparameters}

Table \ref{tab:exp_details} provides Implementation details of the deep models. All experiments are conducted on 24 NVIDIA RTX 2080 Ti GPUs.

\begin{table}[h]
\small
\centering
\begin{tabular}{ p{4.8cm} p{3.6cm} p{3.8cm} p{3.6cm} }
    \toprule
      & \textbf{MNIST}  & \textbf{CIFAR-10} & \textbf{CIFAR-100}\\ 
      & \cite{mnist10} & \cite{cifar10} & \cite{cifar10} \\ 
      \midrule
      \tabincell{l}{Input normalization parameters \\ (normalized = (input-mean)/std)} & \tabincell{l}{ mean: (0.50, 0.50, 0.50) \\ std: (0.50, 0.50, 0.50)} & \tabincell{l}{ mean: (0.49, 0.48, 0.45) \\ std: (0.20, 0.20, 0.20)} & \tabincell{l}{ mean: (0.49, 0.48, 0.45) \\ std: (0.20, 0.20, 0.20)} \\\midrule
    Number of Epochs & 30 & 200 & 200\\\midrule
    Network architecture & ResNet-18 & ResNet-18 & NasNet\\
    & \cite{hzrs16} & \cite{hzrs16} & \cite{zoph2018learning} \\
    \midrule
    Optimizer & \multicolumn{3}{c}{SGD  (momentum = 0.9) \cite{q99}}  \\\midrule
    Initial learning rate (vanilla training) & 0.1 & 0.1 & 0.1 \\\midrule
    Initial learning rate ({\instahide}) & 0.1 & 0.1 & 0.1 \\\midrule
    Learning rate decay & by a factor of 0.1 at the 15$^{th}$ epoch &  by a factor of 0.1 at 100$^{th}$ and 150$^{th}$ epochs & by a factor of 0.2 at 60$^{th}$, 120$^{th}$ and 160$^{th}$ epochs\\\midrule
    Regularization & \multicolumn{3}{c}{$\ell_2$-regularization ($10^{-4}$) } \\\midrule
    Batch size & \multicolumn{3}{c}{128}\\
    \bottomrule
  \end{tabular}
  \caption{Implementation details of network architectures and training schemes. }\label{tab:exp_details}
\end{table}

\subsection{Details of attacks}
We hereby provide more details for the attacks in Section \ref{sec:exp_attack_gradmatch}.

\paragraph{Gradients matching attack.}
Algorithm \ref{alg:attack_grad} describes the gradients matching attack \cite{zlh19}. This attack aims to recover the original image from model gradients computed on it. In the InstaHide setting, the goal becomes to recover $\tilde{x}$, the image after InstaHside. As we have shown in Section~\ref{sec:privacy}, the upper bound on the privacy loss in gradients matching attack is the loss when attacker is given $\tilde{x}$.

\begin{algorithm}[t]
\caption{Gradients matching attack}
  \begin{algorithmic}[1]
  \State {\bf Require : }
  \State The function $F(x; W)$ can be thought of as a neural network, for each $l \in [ L ]$, we define $W_l \in \R^{m_l \times m_{l-1}}$ to be the weight matrix in $l$-th layer, and $m_0 = d_i$ and $m_l = d_o$. $W = \{W_1, W_2, \cdots, W_L\}$ denotes the weights over all layers.
  \State Let $(x_0, y_0)$ denote a private (image, label) pair.
  \State Let $\mathcal{L} : \R^{d_o \times d_o} \rightarrow \R$ denote loss function
   \State Let $g(x,y) = \nabla {\cal L} ( F(x;W) ,y)$ denote the gradients of loss function, and $\hat{g} =  g(x,y) |_{ x =x_0, y= y_0 }$ is the gradients computed on $x_0$ with label $y_0$
  \Procedure{\textsc{InputRecoveryfromGradients}}{$ $}
      \State $x^{(1)} \gets \mathcal{N} (0, 1)$, $y^{(1)} \gets \mathcal{N} (0, 1)$
      \Comment{Random initialization of the input and the label}
      \For{$t = 1 \rightarrow T$}
     \State Let $D_g(x,y) = \|g(x,y) - \hat{g} \|_2^2$
      \State $x^{(t+1)} \gets x^{(t)} - \eta \cdot \nabla_{x} D_g(x,y) |_{x = x^{(t)}} $
      \State $y^{(t+1)} \gets y^{(t)} - \eta \cdot \nabla_{y} D_g(x,y) |_{ y=y^{(t)} } $
      \EndFor
  \State \Return $x^{(T+1)}$, $y^{(T+1)}$
  \EndProcedure
  \end{algorithmic}
  \label{alg:attack_grad}
  \end{algorithm}

\subsection{Results of the Kolmogorov–Smirnov Test}

In order to further understand whether there is significant difference among distributions of {\instahide} encryptions of different $x$’, we run the Kolmogorov-Smirnov (KS) test \cite{k33,s48}.

Specifically, we randomly pick 10 different private $x_i$'s, $i \in [10]$, and generate 400 encryptions for each $x_i$ (4,000 in total). We sample $\tilde{x}_{ij}, j \in [400]$, the encryption of a given $x_i$, and run KS-test under two different settings:
\begin{itemize}
    \item $\tilde{x}_{ij}$ v.s. all encryptions (All):  KS-test(statistics of $\tilde{x}_{ij}$, statistics of all $\tilde{x}_{uj}$’s for $u \in [10]$)
    \item $\tilde{x}_{ij}$ v.s. encryptions of other $x_u$'s, $u \neq i$ (Other): KS-test(statistics of $\tilde{x}_{ij}$, statistics of all $\tilde{x}_{uj}$’s for $u \in [10], u \neq i$)
\end{itemize}

For each $x_i$, we run the Kolmogorov–Smirnov test for 50 independent $\tilde{x}_{ij}$’s, and report the averaged p-value as below (a higher p-value indicates a higher probability that $\tilde{x}_{ij}$ comes from the tested distribution). We test 7 different statistics: pixel-wise mean, pixel-wise standard deviation, total variation, and the pixel values of 4 random locations. KS-test results suggest that, there is no significant differences among distribution of encryptions of different images.

\begin{table}[h]
\scriptsize
    \centering
    \begin{tabular}{ccccccccccccccc}
    \toprule
         & \multicolumn{2}{c}{{\bf Pixel-wise Mean}} & \multicolumn{2}{c}{{\bf Pixel-wise Std}} & \multicolumn{2}{c}{{\bf Total Variation}} & \multicolumn{2}{c}{{\bf Location 1}} & \multicolumn{2}{c}{{\bf Location 2}} & \multicolumn{2}{c}{{\bf Location 3}} & \multicolumn{2}{c}{{\bf Location 4}}\\
         & All & Other & All & Other & All & Other& All & Other& All & Other& All & Other& All & Other\\\midrule
         $x_1$ & 0.49&0.49 &0.54 &0.54&0.54 &0.54&0.61 &0.61&0.58 &0.58&0.61 &0.61&0.59 &0.59\\
         $x_2$ & 0.53&0.52 &0.49 &0.49&0.49 &0.49&0.65 &0.65&0.56 &0.56&0.61 &0.61&0.74 &0.74\\
         $x_3$ & 0.54&0.53 &0.55 &0.55&0.55 &0.55&0.61 &0.61&0.68 &0.68&0.69 &0.69&0.58 &0.58\\
         $x_4$ & 0.49&0.48 &0.49 &0.49&0.49 &0.49&0.48 &0.48&0.40 &0.41&0.70 &0.70&0.70 &0.70\\
         $x_5$ & 0.51&0.51 &0.50 &0.50&0.50 &0.50&0.60 &0.60&0.72 &0.72&0.66 &0.66&0.50 &0.50\\
         $x_6$ & 0.44&0.43 &0.52 &0.51&0.52 &0.51&0.66 &0.66&0.69 &0.69&0.52 &0.52&0.60 &0.59\\
         $x_7$ & 0.45&0.45 &0.48 &0.48&0.48 &0.47&0.62 &0.62&0.52 &0.52&0.64 &0.65&0.67 &0.66\\
         $x_8$ & 0.46&0.46 &0.50 &0.49&0.50 &0.49&0.75 &0.76&0.69 &0.69&0.63 &0.63&0.73 &0.73\\
         $x_9$ & 0.53&0.53 &0.53 &0.53&0.53 &0.53&0.60 &0.60&0.67 &0.67&0.54 &0.54&0.59 &0.59\\
         $x_{10}$ & 0.44&0.44 &0.37 &0.37&0.38 &0.37&0.66 &0.67&0.65 &0.65&0.67 &0.67&0.65 &0.65\\
    \bottomrule
    \end{tabular}
    \caption{Averaged $p$-values for running KS-test on 50 encryptions for each $x_i, i \in [10]$. For each row, `All' and `Other' tests give similar $p$-values, suggesting there is no significant differences among distribution of encryptions of different images.}
    \label{tab:averaged_p_values}
\end{table}